\documentclass[a4paper]{article}

\usepackage{graphicx,epsfig}
\usepackage{amsfonts} 
\usepackage{amsthm, amsmath}
\usepackage{bbm}
\usepackage{indentfirst}  
\usepackage{listings}
\usepackage{caption}
\usepackage{subfig}
\usepackage{epstopdf}
\usepackage{stackengine}
\usepackage{multirow}
\usepackage{bbm,amsmath}
\usepackage{textcomp}
\usepackage{titling}
\usepackage{verbatim}
\usepackage{wrapfig}
\usepackage{float}
\usepackage{csquotes}
\usepackage{tabularx, booktabs}
\usepackage{array}
\newcolumntype{P}[1]{>{\centering\arraybackslash}p{#1}}
\usepackage{authblk}
\usepackage[numbers]{natbib}

\def\reals{\mathbb{R}}

\def\minimize{\mathop{\rm minimize}\limits}

\def\st{\mathop{\rm subject\ to}}

\newtheorem{theorem}{Theorem}

\newtheorem{proposition}{Proposition}

\usepackage{tabularx} 
\usepackage{amssymb}
\usepackage{amsmath}
\usepackage{amsthm}
\usepackage{latexsym}
\usepackage{amsfonts}

\newcommand{\N}{\mathcal{N}}

\let\phi\varphi
\let\epsilon\varepsilon

\begin{document}

\title{Optimal semi-static hedging in illiquid markets}
\author{Teemu Pennanen\thanks{teemu.pennanen@kcl.ac.uk}, Udomsak Rakwongwan\thanks{udomsak.rakwongwan@kcl.ac.uk}}
\affil{Department of Mathematics, \\King's College London, \\Strand, London, WC2R 2LS, United Kingdom}

\maketitle

\begin{abstract}
We study indifference pricing of exotic derivatives by using hedging strategies that take static positions in quoted derivatives but trade the underlying and cash dynamically over time. We use real quotes that come with bid-ask spreads and finite quantities. Galerkin method and integration quadratures are used to approximate the hedging problem by a finite dimensional convex optimization problem which is solved by an interior point method. The techniques are extended also to situations where the underlying is subject to bid-ask spreads. As an illustration, we compute indifference prices of path-dependent options written on S\&P500 index. Semi-static hedging improves considerably on the purely static options strategy as well as dynamic trading without options. The indifference prices make good economic sense even in the presence of arbitrage opportunities that are found when the underlying is assumed perfectly liquid. When transaction costs are introduced, the arbitrage opportunities vanish but the indifference prices remain almost unchanged.
\end{abstract}
\newpage

\section{Introduction}

Unlike in complete markets where derivative prices are uniquely determined by replication arguments, in incomplete markets, quoted prices depend on subjective factors such as the agents' financial positions, risk preferences and views concerning future market developments. Such dependencies are consistently described by indifference pricing which can be viewed as an extension of replication arguments to the incomplete markets; see e.g.\ \cite{buh70,hn89,car9,ijs04} and the references therein. Extensions to illiquid markets and the corresponding duality theory has been studied in \cite{pen14} and \cite{pp18e}, respectively.
 
This paper develops computational techniques for utility-based semi-static hedging with a finite number of derivatives whose quotes have bid-ask spreads and finite quantities. The hedging strategies involve buy-and-hold positions in the derivatives while the underlying and cash are traded dynamically. We use a Galerkin method to approximate the hedging problem by a finite-dimensional convex optimization problem which is then numerically solved by an interior point method much like in \cite{apr18} in a purely static setting. The approach extends with minor modifications to situations where the dynamically traded underlying is also subject to bid ask spreads.

The techniques are illustrated by computing indifference prices of various path-dependent options (including knock-out, Asian and look-back options) on the S\&P500 index. As hedging instruments, we use exchange-traded puts and calls on the index. For the nearest maturities, one can find hundreds of options with bid and/or ask quotes. 
We find that semi-static hedging significantly improves on the hedges obtained by purely static or purely dynamic strategies. The semi-static hedging strategies provide good approximations of the payouts of the hedged derivatives and the corresponding spreads between seller's and buyer's prices are considerably tighter than those obtained with purely static or dynamic hedging. The computational approach applies to arbitrary utility functions and stochastic models that allow for numerical sampling.

Compared to the more traditional super/subhedging, indifference pricing is less sensitive to market imperfections and it makes good sense even in the presence of arbitrage. This was found a useful feature as the quotes on exchange traded options seem to often lead to arbitrage if the dynamically traded underlying is assumed perfectly liquid (as is the case in most models studied in the literature). The arbitrage opportunities vanish when moderate transaction costs on the underlying are introduced but the indifference prices remain almost unaffected. 

When the statically hedged options are discarded, the optimal investment strategy to maximize the expected exponential utility coincides with the classical Merton strategy. More surprisingly, the corresponding hedging strategies obtained with indifference pricing seem to coincide with the delta-hedging strategies for replication in complete market models.

Semi-static hedging has been actively studied in the recent literature but mainly under the assumption of perfect liquidity for both the static and dynamic instruments. Moreover, it is common to assume also that there exist quotes for a continuum of strikes as opposed to the finite number of strikes traded in real markets. Much of the research has focused on duality theory in a distributionally roust superhedging; see e.g.\ \cite{bhp13}. Guo and Obloj~\cite{go19} developed computational techniques for the martingale optimal transport problems by using discretization and interior point methods much like we do below. Their problem can be viewed as the dual of a semi-static superhedging problem with a continuum of strikes for statically held call options (which fixes the marginals of the martingale measures). Extensions of the duality theory of model-free semi-static superhedging to illiquid markets were given in the examples of \cite{pp19}. In the computational studies below, we find that with real finite-liquidity quotes for finitely many options, superhedging tends to give prices with very wide spreads.

The present paper is closely related to \cite{is06} and \cite{ijs08} that studied utility indifference pricing under semi-static trading. While they studied duality and asymptotics of indifference prices in a perfectly liquid continuous-time model, we focus on real illiquid markets and compute prices and hedging portfolios numerically.



The rest of this paper is organized as follows. Sections~\ref{sec:po} and \ref{sec:idp} describe the optimal hedging model and the corresponding indifference prices, respectively. Section~\ref{sec:numopt} presents the techniques employed in the numerical computation of optimal hedging and indifference pricing. Section~\ref{sec:illiquid} extends the techniques to markets with an illiquid underlying. Sections~\ref{sec:num} and \ref{sec:idpnum} present the numerical results obtained with S\&P500 derivatives.

\section{Semi-static optimisation}\label{sec:po}


Consider a finite set $J$ of quoted derivatives whose payouts are determined by the values of an underlying index $X$ at times $t=0,1,\ldots,T$. We assume that the derivatives are traded only at $t=0$ and they are held to maturity. The underlying can be traded at any $t=0,1,\ldots,T$. The cost of buying $x^j$ units of $j\in J$ is denoted by
\begin{equation*}
S^j_0(x^j) :=
\begin{cases}
s^j_a x^j & \text{if $x^j\ge 0$},\\
s^j_b x^j & \text{if $x^j\le 0$},
\end{cases}
\end{equation*}
where $s^j_b< s^j_a$ are the bid and ask prices of $j$. The quantities available at the best bid and ask quotes will be denoted by $q^j_b$ and $q^j_a$, respectively. This means that the position $x^j$ we take in asset $j$ has to lie in the interval $[-q^j_b,q^j_a]$. The payoff of $j\in J$ at time $t$ will be denoted by $P_t^j$. We assume for now that the underlying index is perfectly liquid and can be dynamically traded at price $X_t$, $t=0,\ldots,T$. More realistic markets will be considered in Section~\ref{sec:illiquid}.

Consider an agent with $w\in\reals$ units of initial cash and the liability to deliver $c_t$ units of cash at $t=1,\ldots,T$. In the applications below, $c_t$ will be the payout of an exotic option to be priced. We model the price process $X=(X_t)_{t=1}^T$, the payouts $p^j=(p_t^j)_{t=1}^T$ and the liability $c=(c_t)_{t=1}^T$ as adapted stochastic processes in a filtered probability space $(\Omega,\mathcal{F},(\mathcal{F}_t)_{t=1}^T,P)$. We will study the optimal investment problem
\begin{equation}\label{alm}\tag{SSP}
\begin{aligned}
&\minimize & & Ev\left(\sum_{t=1}^{T} [c_t-\sum_{j\in J}p_t^jx^j]-\sum_{t=0}^{T-1}z_t\Delta X_{t+1}\right) \ \text{over} \ x\in D, z\in\N\\ 
&\st & &\sum_{j\in J}  S^j_0 (x^j) \leq w,
\end{aligned}
\end{equation}
where 
\begin{equation*}
D:=\prod_{j\in J}[-q^j_b,q^j_a],
\end{equation*}
$\mathcal{N}$ is the linear space of adapted trading strategies $z=(z_t)_{t=1}^{T-1}$, and $v:\reals\to\reals$ is a {\em loss function} describing the investor's risk preferences; see e.g.\ \cite[Section~4.9]{fs11}. One may think of $u(c):=-v(-c)$ as a utility function so $v$ will be assumed nondecreasing and convex. The argument of $v$ is the unhedged part of the claims ${(c_t)}^{T}_{t=1}$, the last term being interpreted as the payout of a self-financing trading strategy in the underlying and cash. One could also include various margin requirements in the specification of the set $D$. 

It is clear that the optimum value and solutions of problem \eqref{alm} depend on
\begin{enumerate}
\item
the financial position described by the initial cash $w$ and liability $c$,
\item
the views concerning the future values of $X$, $p$ and $c$ described by the probabilistic model,
\item
the risk preferences described by the loss function $v$
\end{enumerate}
all of which are subjective. The effect of these factors on the optimal hedging strategies and the associated prices of $c$ will be studied below. It turns out that, if the claims $c$ are replicable, then the prices will be unique and independent of the subjectivities; see Theorem~\ref{thm:ab} below.

Another important feature of \eqref{alm} is that it is a {\em convex optimization problem}. Convexity is crucial in numerical solution of \eqref{alm} as well as in the mathematical analysis of the indifference prices.

\section{Indifference pricing}\label{sec:idp}

We shall denote the optimum value of \eqref{alm} by
\begin{multline*}
\varphi(w,c):=\inf_{x\in D,\ z\in\N}\bigg\{Ev\big(\sum_{t=1}^{T} [c_t-\sum_{j\in J}p_t^j x^j]-\sum_{t=1}^{T-1}\Delta X_{t+1} z_t\big)
\bigg|\, \sum_{j\in J} S^j_0 (x^j) \le w\bigg\}.
\end{multline*}
For an agent with financial position of $\bar w$ units of initial wealth and a liability of delivering a sequence $\bar c = (\bar c_t)_{t=1}^{T}$ of payments, the {\em indifference price} for selling a claim $c=(c_t)_{t=1}^T$ is given by
\begin{equation*}
\pi_s(\bar{w},\bar{c};c):=\inf\{w\in\reals\,|\,\varphi(\bar{w}+w,\bar{c}+c)\leq\varphi(\bar{w},\bar{c})\}.
\end{equation*}
This is the minimum price at which the agent could sell the claim $c$ without worsening her financial position as measured by the optimum value of \eqref{alm}. Analogously, the indifference price for buying $c$ is given by
\begin{equation*}
\pi_b(\bar{w},\bar{c};c):=\sup\{w\in\reals\,|\,\varphi(\bar{w}-w,\bar{c}-c)\leq\varphi(\bar{w},\bar{c})\}.
\end{equation*}

We shall compare the indifference prices with the superhedging and subhedging costs defined by
\begin{multline*}
\pi_{\sup}(c):=\inf_{x\in D,\ z\in\N}\bigg\{\sum_{j\in J} S^j_0 (x^j)\,\bigg|\, 
\sum_{t=1}^T \sum_{j\in J}p_t^j x^j +\sum_{t=1}^{T-1}\Delta X_{t+1} z_t-\sum_{t=1}^{T} c_t \ge 0\ P\text{-a.s.}\bigg\},
\end{multline*}
and
\begin{multline*}
\pi_{\inf}(c):=\sup_{x\in D,\ z\in\N}\bigg\{-\sum_{j\in J} S^j_0 (x^j)\,\bigg|\,
\sum_{t=1}^T \sum_{j\in J}p_t^j x^j +\sum_{t=1}^{T-1}\Delta X_{t+1} z_t+\sum_{t=1}^{T} c_t \ge 0\ P\text{-a.s.}\bigg\}.
\end{multline*}
The superhedging cost is the least cost of a superhedging portfolio while the subhedging cost is the greatest revenue one could get by entering position that superhedges the negative of $c$. Whereas the indifference prices of a claim depend on our financial position, views and risk preferences described by $(\bar w,\bar c)$, $P$ and $v$, respectively, the superhedging and subhedging costs are independent of such subjective factors. 

In situations where the quantities available at the best quotes are large enough to be nonbinding, the indifference prices lie between the superhedging and subhedging costs. Indeed, an application of \cite[Theorem~4.1]{pen14} to the present situation gives the following.

\begin{theorem}\label{thm:ab}
The function $\pi_s(\bar w,\bar c;\cdot)$ is convex, nondecreasing and $\pi_s(\bar w,\bar c;0)\le 0$. If there are no quantity constraints (or if they are not active), then $\pi_s(\bar w, \bar c;c) \le \pi_{\sup}(c)$. If in addition, $\pi_s(\bar w, \bar c;0)=0$, then
\[
\pi_{\inf}(c)\le\pi_l(\bar w, \bar c;c) \le \pi_s(\bar w, \bar c;c) \le \pi_{\sup}(c)\quad\forall c\in L^0
\]
with equalities throughout if $s^b=s^a$ and $c$ is replicable.
\end{theorem}

As long as one can (numerically) compute the optimum values $\varphi(w,c)$ for given $(w,c)$, the indifference prices can be computed by a simple line search with respect to the price. This can, however, be computationally expensive. If cash is perfectly liquid and the interest rate is zero, the indifference prices can be expressed in terms of two optimization problems as follows.

\begin{proposition}
If cash is perfectly liquid with zero interest rate, the indifference prices for buying and selling for an agent with exponential risk measure can be expressed as,
$$\pi_b(\bar w, \bar c;c) =\frac{\bar{w}}{\lambda}\log\bigg(\frac{\varphi(\bar w,\bar c)}{\varphi(\bar w,\bar c-c)}\bigg),$$
$$\pi_s(\bar w, \bar c;c) =\frac{\bar{w}}{\lambda}\log \bigg(\frac{ \varphi(\bar w,\bar c+c)}{ \varphi(\bar w,\bar c)}\bigg),$$
where $\bar{w}$ is an initial wealth, and $\lambda$ is a risk aversion factor.
\end{proposition}

\begin{proof}
By definition,
\begin{equation*}
\begin{split}
\varphi(\bar{w}+w,\bar{c}+c)=&\inf_{x\in D,\ z\in\N}\bigg\{E\exp(\frac{\lambda}{\bar{w}}(\sum_{t=1}^{T} [\bar{c}_t+c_t-\sum_{j\in J}p^j_t x^j]-\sum_{t=1}^{T-1}\Delta X_{t+1} z_t))\,\\
&\qquad\qquad\bigg|\, \sum_{j\in J}  S^j_0 (x^j) \leq \bar{w}+w\bigg\},\\
=&\inf_{x\in D,\ z\in\N}\bigg\{E\exp(\frac{\lambda}{\bar{w}}(\sum_{t=1}^{T} [\bar{c}_t+c_t-\sum_{j\in J}p^j_t x^j]-w-\sum_{t=1}^{T-1}\Delta X_{t+1} z_t))\\
&\qquad\qquad\bigg|\, \sum_{j\in J}  S^j_0 (x^j) \leq \bar{w}\bigg\},\\
=&\varphi(\bar{w},\bar{c}+c)\exp(\frac{-\lambda w}{\bar{w}}).
\end{split}
\end{equation*}
Thus,
\begin{equation*}
\begin{split}
\pi_b(\bar{w},\bar{c};c)=&\inf\{w\,|\,\varphi(\bar{w}-w,\bar{c}-c)\leq\varphi(\bar{w},\bar{c})\},\\
=&\inf\{w\,|\,\varphi(\bar{w},\bar{c}-c)\exp(\frac{\lambda w}{\bar{w}})\leq\varphi(\bar{w},\bar{c})\},\\
=&\frac{\bar{w}}{\lambda}\log\bigg(\frac{\varphi(\bar w,\bar c)}{\varphi(\bar w,\bar c-c)}\bigg),\\
\end{split}
\end{equation*}
and 
\begin{equation*}
\begin{split}
\pi_s(\bar{w},\bar{c};c)=&\inf\{w\,|\,\varphi(\bar{w}+w,\bar{c}+c)\leq\varphi(\bar{w},\bar{c})\},\\
=&\inf\{w\,|\,\varphi(\bar{w},\bar{c}+c)\exp(\frac{-\lambda w}{\bar{w}})\leq\varphi(\bar{w},\bar{c})\},\\
=&\frac{\bar{w}}{\lambda}\log \bigg(\frac{ \varphi(\bar w,\bar c+c)}{ \varphi(\bar w,\bar c)}\bigg),
\end{split}
\end{equation*}
which completes the proof.
\end{proof}

\section{Numerical optimization} \label{sec:numopt}

We assume from now on that the derivative and liability payouts $p^j$ and $c$, respectively, are adapted to the filtration generated by the underlying $X$. This clearly holds when $p^j$ and $c$ are contingent claims on $X$.

\subsection{Galerkin method} \label{galerkinSSP}

To solve \eqref{alm}, we need to optimize the dynamic part $z$ over the infinite-dimensional space $\N$ of adapted stochastic processes. We will employ the {\em Galerkin method} where one optimizes $z$ only over the finite-dimensional subspace $\N^N\subset\N$ spanned by the simple processes $z^{s,n}\in\N$ of the form
\begin{equation*}
z^{s,n}_t(\omega):=
\begin{cases}
1 & \text{if $t=s$, $X_t(\omega)\in [K_n,K_{n+1})$},\\
0 & \text{otherwise},
\end{cases}
\end{equation*}
where $s=1,2,\ldots,T-1$ and $K_n$, $n=1,2,\ldots,N_s$ are the strikes of the quoted options with maturity $s$, while $K_0=0$ and $K_{N_s+1}=+\infty$. The dimension of the linear span $\N^N$ is thus $N=\prod_{s=1}^{T-1}(N_s+1)$. Clearly, each $z^{s,n}_t$ is adapted to the filtration generated by $X$ so indeed, $\N^N\subset\N$. The linear span $\N^N$ consists of simple processes that that are constant between consecutive strikes.

\subsection{Integration quadrature} \label{Integration quadrature}

Since the filtration $(\mathcal{F}_t)_{t=0}^T$ is generated by $X$, the Doob-Dynkin lemma implies that the random variables $c_t$ and $p_t$ are functions of $X^t=(X_1,X_2,\ldots,X_t)$. The objective of \eqref{alm} can be written as 
$$Ef(x,z)=\int_{\reals_+^T} f(x,z(X),X) \phi (X)dX,$$
where $\phi $ is the density function of X, and
\[
f(x,z(X),X) := v\left(\sum_{t=1}^{T} [c_t(X^t)-\sum_{j\in J}p_t^j (X^t) x^j]-\sum_{t=1}^{T-1}\Delta X_{t+1} z_t (X_t)\right).
\]

We will approximate the multivariate integral by an integration quadrature:
\[
\int_{\reals_+^T} f(x,z(X),X) \phi (X)dX \approx \sum_{i=1}^{M} w_i f(x,z(X^{(i)}),X^{(i)}) \phi(X^{(i)}),
\]
where $M$ is the number of the quadrature points, $X^{(i)}$ are the quadrature points and $w_i$ are the corresponding weights.

There are many possible choices for the integration quadrature. In this study, we take $X^{(i)}=(X^{(i)}_t)_{t=1}^T$ where $X^{(i)}_t$ are the strikes at maturity $t$. The corresponding weights $w_i$ will be the volumes of the hyper cubes defined by the consecutive strikes. This choice of quadrature points yields fairly accurate results because the portfolio payout depends linearly on the index value between two strikes. In addition, the probability that the index is smaller than the smallest strike or bigger than the biggest strike is very small.

\subsection{Interior point method} \label{interiorSSP}

The approximate problem obtained with the Galerkin method and the integration quadrature, problem \eqref{alm} becomes a finite-dimensional convex optimization problem with finitely many constraints. It can be written as
\begin{align*}
  &\minimize\quad\sum_{i=1}^M v\left(\sum_{t=1}^{T} [c_t(X^t)-\sum_{j\in J}p_t^j (X^t) x^j]-\sum_{t=1}^{T-1}\Delta X_{t+1} z_t (X_t)\right)w_i\phi(X^i) \\
  &\text{over}\quad x\in D,\ z\in\N^N \\
&\st\quad \sum_{j\in J}  S^j_0 (x^j) \leq w.
\end{align*}
This is a finite-dimensional convex optimization problem that can be solved numerically e.g.\ by interior-point methods. In this study, we use the exponential utility so the problem can be written as a conic exponential optimization problem and solved using the interior-point solver of MOSEK~\cite{mosek}. Numerical results are given in Section~\ref{sec:num} below.

\section{Extension to illiquid underlying}\label{sec:illiquid}

Up to now, we have assumed that the underlying index is a perfectly liquid asset that can bought and sold at price $X$. The same assumption is made in most of the literature on semi-static hedging but from a practical point of view, this is not quite realistic. This section considers a more realistic variant of \eqref{alm} where the index is subject to a transaction costs, or equivalently, a constant bid-ask spread. More precisely, we assume that an agent needs to pay a $\delta$\% transaction cost to buy or sell the index at time $t=1,2,\ldots,T-1$, and the index is liquid at $T$. The semi-static hedging problem can then be written as
\begin{equation*} 
\begin{aligned}
&\minimize\quad& &Ev\left(\sum_{t=1}^{T} [c_t-\sum_{j\in J}p_t^j x^j]+\sum_{t=1}^T S_t(\Delta z_t)\right) \quad\text{over} \ x\in D, z \in\N\\ 
&\st\quad& &\sum_{j\in J}  S^j_0 (x^j) \leq w,
\end{aligned}
\end{equation*}
where $\Delta z_t:=z_t-z_{t-1}$ is the number of the unit of the underlying bought at $t$ and
\begin{equation*}
S_t(\Delta z_t) :=
\begin{cases}
(1+\frac{\delta}{100})X_t \Delta z_t & \text{if $\Delta z_t\ge 0$},\\
(1-\frac{\delta}{100})X_t \Delta z_t & \text{if $\Delta z_t\le 0$}.
\end{cases}
\end{equation*}
Here and in what follows, $z_{-1}=z_T=0$. Note that if $\delta=0$,
$$\sum_{t=1}^T S_t(\Delta z_t)=\sum_{t=1}^T X_t\Delta z_t=-\sum_{t=1}^{T-1}\Delta X_{t+1}{z_t},$$
so the original model \eqref{alm} is a special case of the above.

To solve the problem numerically, we express the purchases $\Delta z$ as
\[
\Delta z_t=\Delta z_t^+ -\Delta z_t^- ,
\]
where $\Delta z_t^+, \Delta z_t^-\ge 0$ represent purchases and sales, respectively, of the underlying. Thus, the trading cost can be written as
\[
S_t(\Delta z_t) = (1+\frac{\delta}{100})X_t \Delta z^+_t -(1-\frac{\delta}{100})X_t \Delta z^-_t.
\]
We then can apply the Galerkin method to $\Delta z^+_t$ and  $\Delta z^-_t$ where the multipliers of the basis functions are restricted to be positive. The rest is similar to the numerical solution of \eqref{alm} described in Section~\ref{sec:numopt}. 


\section{An application to S\&P500 derivatives}\label{sec:num}

This section illustrates the presented models and techniques in the S\&P500 derivatives market with option quotes taken from Bloomberg. For the nearest maturities, there are hundreds of exchange traded puts and calls whose quotes come with bid-ask spreads and finite quantities. The resulting optimization and pricing problems are then solved using the techniques described in Sections~\ref{sec:numopt}--\ref{sec:illiquid}.

We start by finding optimal portfolios in the quoted derivatives when assuming that the liability $c$ in Problem~\eqref{alm} is zero. We study the dependence of the optimal solution on the risk preferences as well as on the distribution of the underlying, both of which are highly subjective components of the model. Building on the optimization model, we then compute indifference prices for path-dependent derivatives namely a knock-out call option, an Asian call option, look-back call options and a look-back digital option. We compare the optimized portfolios and indifference prices obtained by semi-static hedging with different transaction costs to those obtained by purely dynamic hedging without options.

\subsection{Quotes, views, and preferences}\label{sec:qvp}

We use quotes for S\&P500 index options with maturities 21 April 2017 and 19 May 2017. The strikes of the options range from 1500 to 2500. The quotes were obtained from Bloomberg on 21 March 2017 at 3:00:00 PM when the value of the S\&P500 index was $2360$. All quotes come with bid ask spreads and finite quantities. Table~\ref{table:BlooembergQuotes} gives an example of quotes on put and call options written on the S\&P500 index available on 21 March 2017 at 15:00:00 expiring on 21 April 2017 and 19 May 2017. The index value was 2,360 at the time. The bid and ask prices shown in the table are per one option, whereas the available quantities are given in terms of a lot size which is 100. For example, the cost of buying or selling a call option with strike 2300 expiring on 5/19/2017 are 81.80 and 79.50, and there are 51 contracts (5100 options) and 48 contracts (4800 options) available for buying and selling respectively.

Call options are more liquid at lower strikes. One can find quotes for call options whose strikes are as low as 500, whereas the lowest strike for put options available in the market is 1,555. For the two nearest maturities, one can find quotes for 678 options whose strikes range from 1,500 to 2,500 with 5 dollar increments. 

\begin{table}[h!]
\centering
 \resizebox{0.95\textwidth}{!}{  
\begin{tabular}{l|c|c|c|c|c}
 \hline 
 Ticker & Type & Bid quantity & Bid price & Ask price & Ask quantity \\
 \hline\hline
 SPX US 5/19/2017 C2300 Index & Call & 48
 & 79.5 & 81.8 & 51\\
 SPX US 5/19/2017 P2300 Index & Put & 182 & 22.6 & 24 & 376\\
 SPX US 4/21/2017 C2370 Index & Call & 300 & 18.7 & 20.3 & 273\\
SPX US 4/21/2017 P2370 Index & Put & 275 & 28.6 & 30.5 & 322\\
 \hline
\end{tabular}
}
\captionsetup{width=1\textwidth}
\caption{Market quotes extracted from Bloomberg on 21 March 2017 at 15:00:00 for put and call options expiring on 21 April 2017 and 19 May 2017.}
\label{table:BlooembergQuotes}
\end{table}

In the applications below, we assume zero interest on cash. In practice, the index is not tradable, but one can trade exchange-traded funds, ETFs, which are securities that track the index. An example of a fairly liquid ETF which efficiently tracks the S\&P500 index is the SPY is issued by State Street Global Advisors.

We model the logarithm of the S\&P500 index by a variance gamma process, obtained by evaluating Brownian motion with drift $\theta$ and volatility $\sigma$ at a random time change given by a gamma process with a unit mean rate and a variance rate $\nu$; see~\cite{mcc98} and~\cite{testt}. The parameters $\theta$ and $\nu$ provide control over skewness and kurtosis, respectively. As a base case, we use the parameter values given in Table \ref{table:parameters}. The parameter $\theta$ is assumed to be zero, whereas the parameters $\sigma$, and $\nu$ are estimated using 10-year historical daily data. The effect of varying the parameters will be studied later on. The initial wealth $w$ is $100,000$ USD and for now, the claim $c_t$ is assumed to be zero for all $t$.

As for risk preferences, we use exponential loss function
\begin{equation*}
v(c) = e^{\lambda c/w},
\end{equation*}
where $w$ is the initial wealth and $\lambda>0$ is a risk aversion parameter. It should be noted that, in general, the net position at maturity can take both positive as well as negative values which prevents the use of utility functions with constant relative risk aversion.

\begin{table}[h!]
  \centering
  \begin{tabular}{c|c|c|c}
\hline
    $\lambda$&$\theta$ & $\sigma$ & $\nu$\\
\hline
    2&0 &   0.1206 & 0.0031  \\
\hline
   \end{tabular}
  \captionsetup{width=1\textwidth}
  \caption{Base-case parameters including a risk aversion factor $\lambda$ and variance gamma parameters $\theta, \sigma$, and $\nu$ used to model the index value.}
  \label{table:parameters}
\end{table}

\subsection{Optimized strategies} \label{optimizedStrategies}

To simplify the presentation and to ease the extensive numerical computations on a relatively modest computational setup, we will study a two-period instance of semi-static hedging and pricing. With the available set of quotes options and the numerical procedure described in Section~\ref{sec:numopt}, there are over 1700 variables and 2700 constraints in the discretized optimization problem. In the quadrature approximation of the expected loss function there are over 160,000 points on the grid; see Section \ref{Integration quadrature}. The interior point solver of MOSEK takes on average of 650.40 seconds on a PC with Intel(R) Core(TM) i5-4690 CPU @ 3.50GHz processor and 16.00 GB memory.

Figure \ref{fig:optPortBaseSSP} represents the structure of the optimized semi-static strategy. The bars represent the optimal positions in the options, whereas the line plots show the positions in cash and the index taken at $t=1$ as functions of $X_1$. Figure \ref{fig:optPayoffBaseSSP} plots the payout of the portfolio as a function of $X_1$ and $X_2$. 

The portfolio enjoys higher profit if the index values at the first and second maturities are close to each other, while its loss is greater elsewhere. This makes sense as it is unlikely that the index value at the second maturity greatly deviates from that of the first maturity.

\begin{figure}[H]
  \centering
    \includegraphics[width=0.9\textwidth,height=0.45\textwidth]{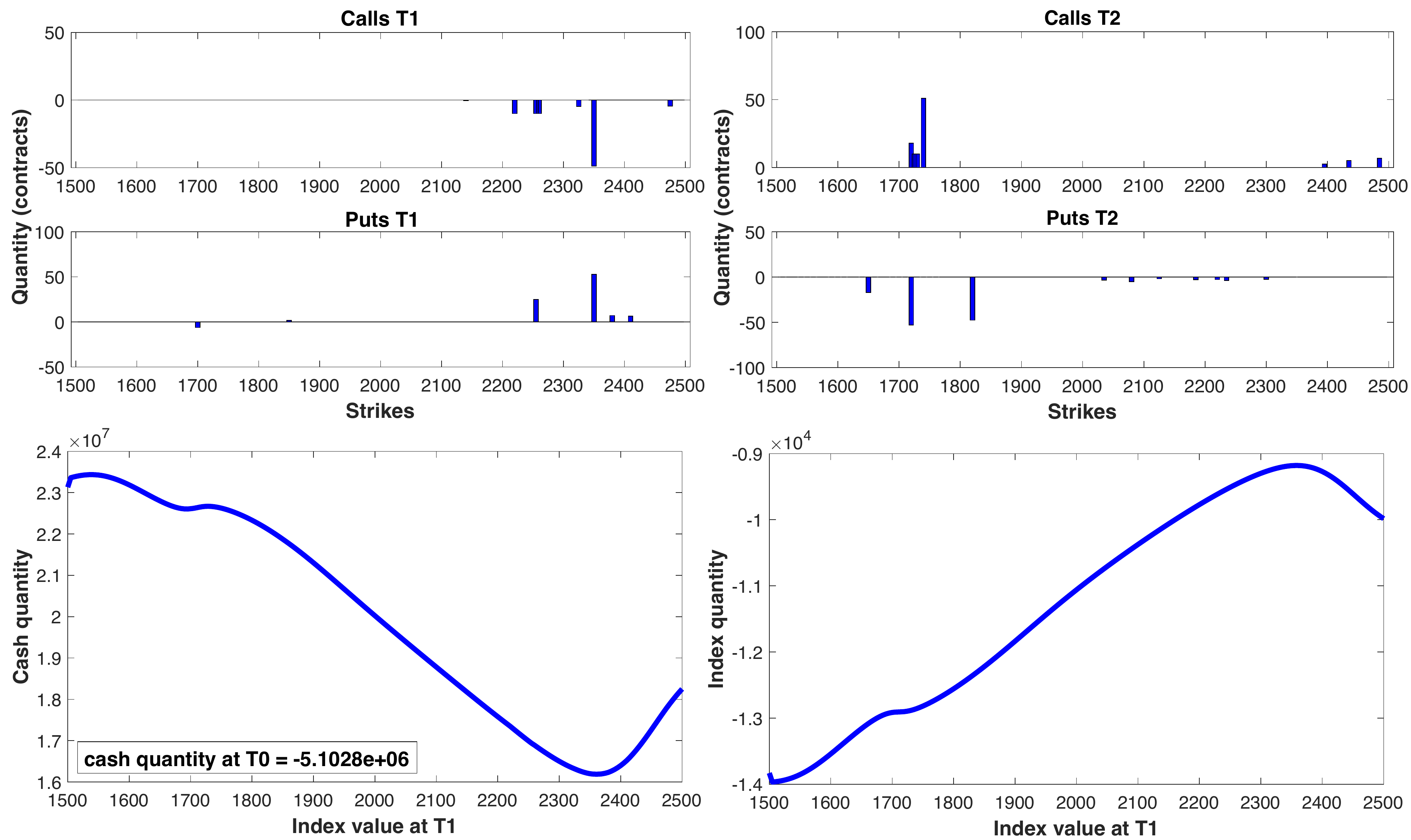}
    \caption{The structure of the optimized semi-static strategy where an index value is modelled by symmetric variance gamma with parameters given in Table \ref{table:parameters}.}\label{fig:optPortBaseSSP}
\end{figure}

\begin{figure}[H]
  \centering
    \includegraphics[width=0.5\textwidth,height=0.32\textwidth]{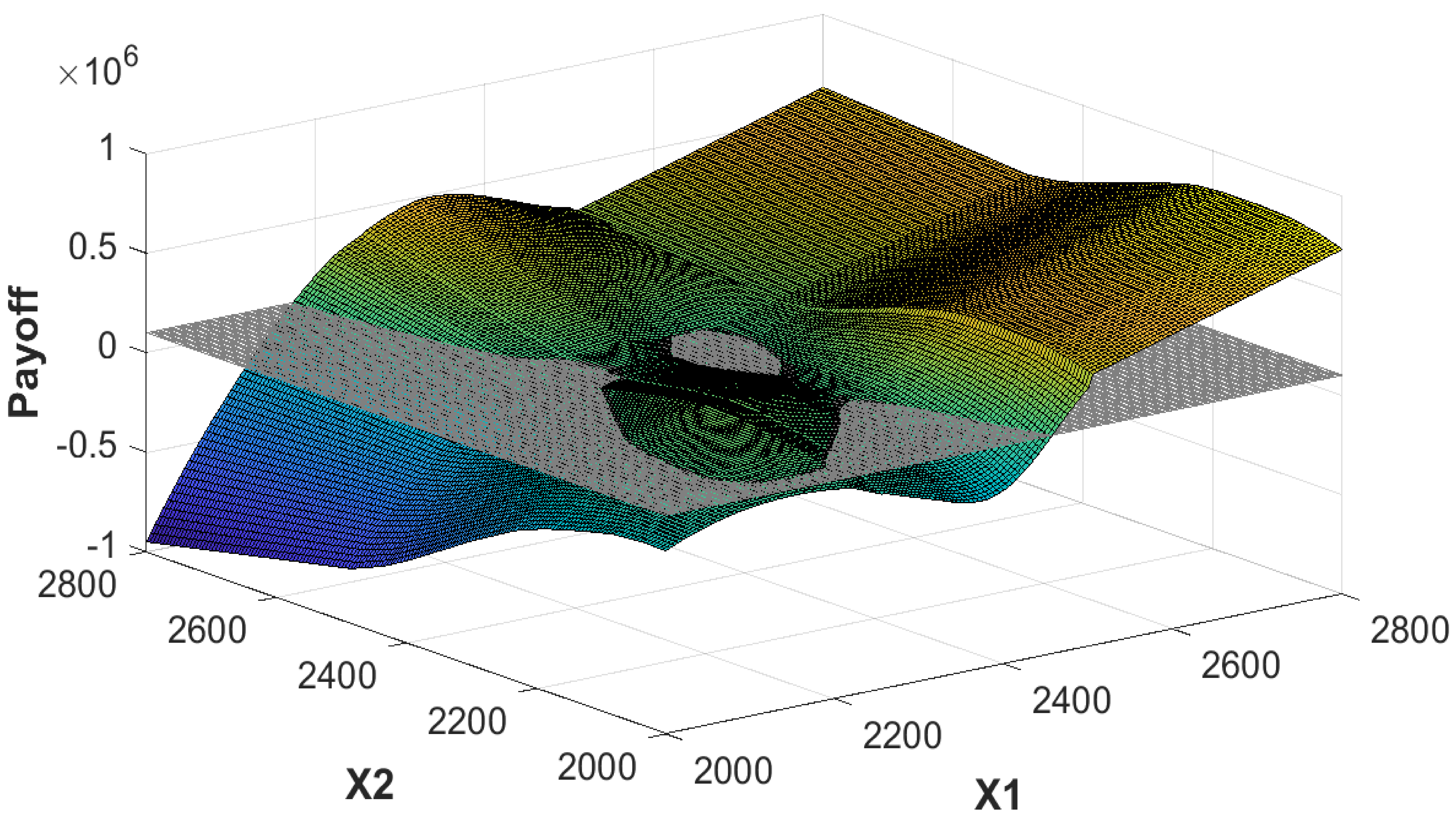}
    \caption{The payout of the optimal portfolio by semi-static optimization as a function of $X_1$ and $X_2$. The grey horizontal plane represents the initial wealth.}\label{fig:optPayoffBaseSSP}
\end{figure}

Figure \ref{fig:optPortRisk6SSP} represents the structure of the optimized portfolio obtained with risk aversion $\lambda =6$. The other parameters are as in Table \ref{table:parameters}. The payout of the portfolio is plotted in Figure \ref{fig:optPayoffDifferentRisk} (left) together with the payout of the optimal portfolio obtained with risk aversion $\lambda=2$. The right plot of Figure \ref{fig:optPayoffDifferentRisk} shows the kernel density estimates (using 1,000,000 out-of-sample simulated price paths) of the terminal wealth of the optimal portfolios obtained with risk aversion 2 and 6.

\begin{figure}[H]
  \centering
    \includegraphics[width=0.9\textwidth,height=0.45\textwidth]{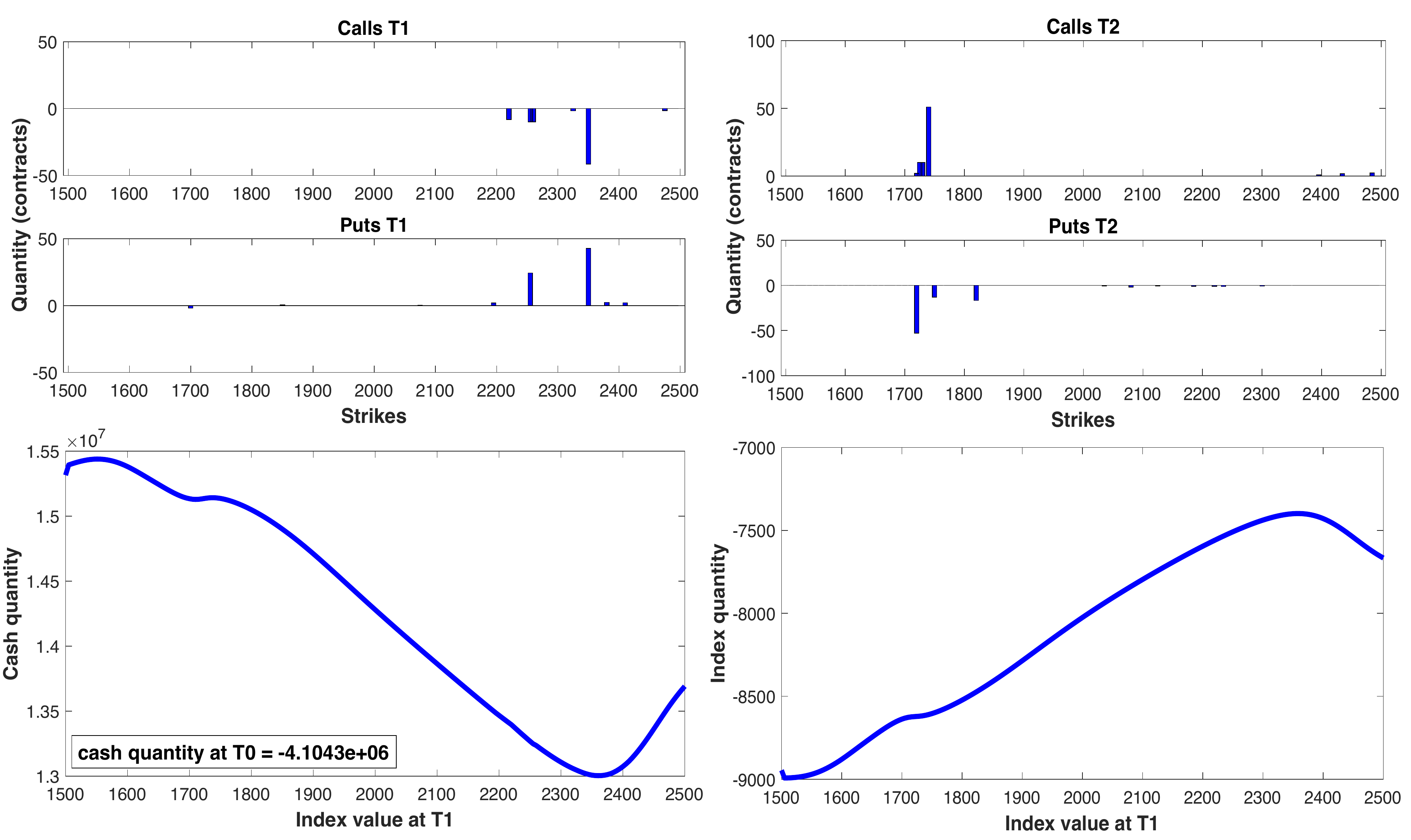}
    \caption{The structure of the optimized semi-static strategy obtained when the risk aversion increased from 2 to 6.}\label{fig:optPortRisk6SSP}
\end{figure}

\begin{figure}[H]
  \centering
    \includegraphics[width=0.9\textwidth,height=0.32\textwidth]{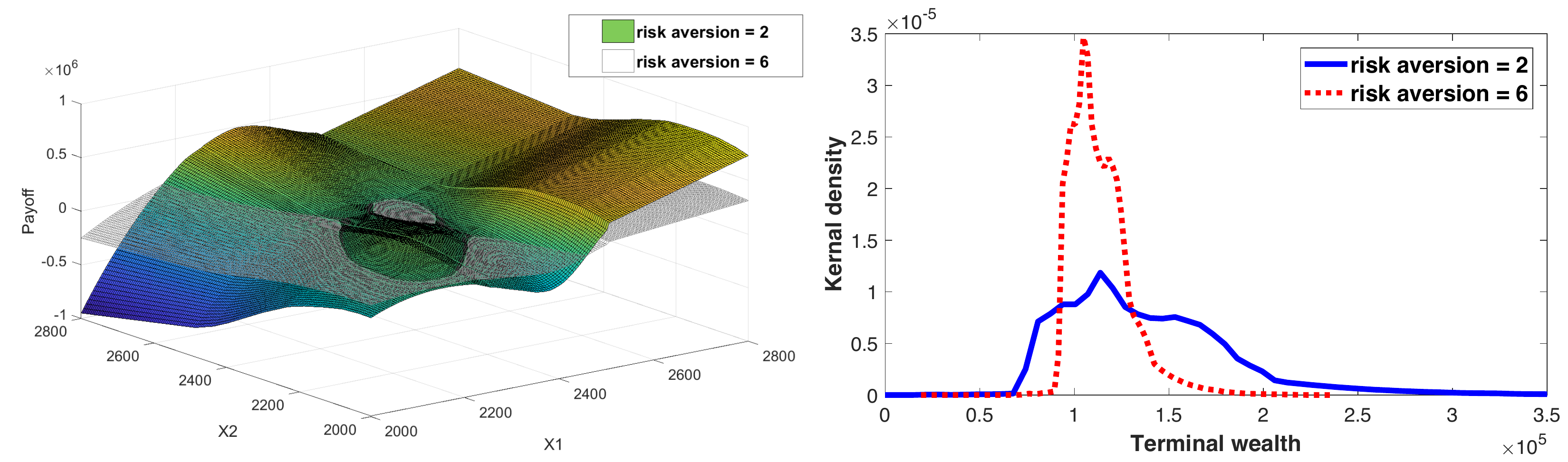}
    \caption{The payout as functions of $X_1$ and $X_2$ of the optimal portfolios obtained with risk aversions 2 and 6 (left). The kernel density estimates of the terminal wealths of the optimal portfolios obtained with risk aversions 2 and 6 using 1,000,000 out-of-sample simulated price paths (right).}\label{fig:optPayoffDifferentRisk}
\end{figure}

We see that the positions of the optimized portfolio obtained with risk aversion $\lambda=6$ are smaller than the ones with risk aversion $\lambda=2$. As expected, the payout of the portfolio obtained with higher risk aversion is less variable. Except for the scale, the shapes of the two kernel density plots look fairly similar, both exhibiting profits in roughly the same area. This makes sense as changing risk aversion does not change the view on the index value.

\begin{table}[h!]
  \centering
  \resizebox{0.55\textwidth}{!}{  
  \begin{tabular}{c|c|c|c}
\hline
    base case & $\lambda = 6$ & $\sigma =0.08$ &$\sigma = 0.2$\\
\hline
    -2.4404 &  -6.4832 & -2.7947 & -2.7895 \\
\hline
   \end{tabular}
   }
  \captionsetup{width=1\textwidth}
  \caption{Logarithms of optimum objective values obtained with different parameters in the semi-static optimization problem~\eqref{alm}. As one of the parameters changes, the others remain the same as in the base case.}
  \label{table:disutilitiesSSP}
\end{table}

Figure~\ref{fig:optPayoffDifferentSigmaSSP} plots the optimal payouts obtained with $\sigma=0.08$ (left) and $\sigma=0.2$ (right). Not surprisingly, increasing $\sigma$ results in a portfolio that gives higher payout further in the tails (a straddle). Table \ref{table:disutilitiesSSP} gives the logarithms of the optimum objective values when $\sigma=0.1206$, $\sigma=0.08$ and $\sigma=0.2$. Note that, since we use the exponential loss function, the logarithmic objective is the ``entropic risk measure'' which has units of cash. It can also be interpreted as the ``certainty equivalent''. The highest objective value is obtained with the base-case parameters which are estimated from historical data. This may be thought of as consistency of the market quotes and the market participants' views of the future behavior of the underlying. With a model that is inconsistent with the ``market views'', the available quotes may seem to offer profitable trading opportunities. 

\begin{figure}[H]
  \centering
    \includegraphics[width=0.9\textwidth,height=0.32\textwidth]{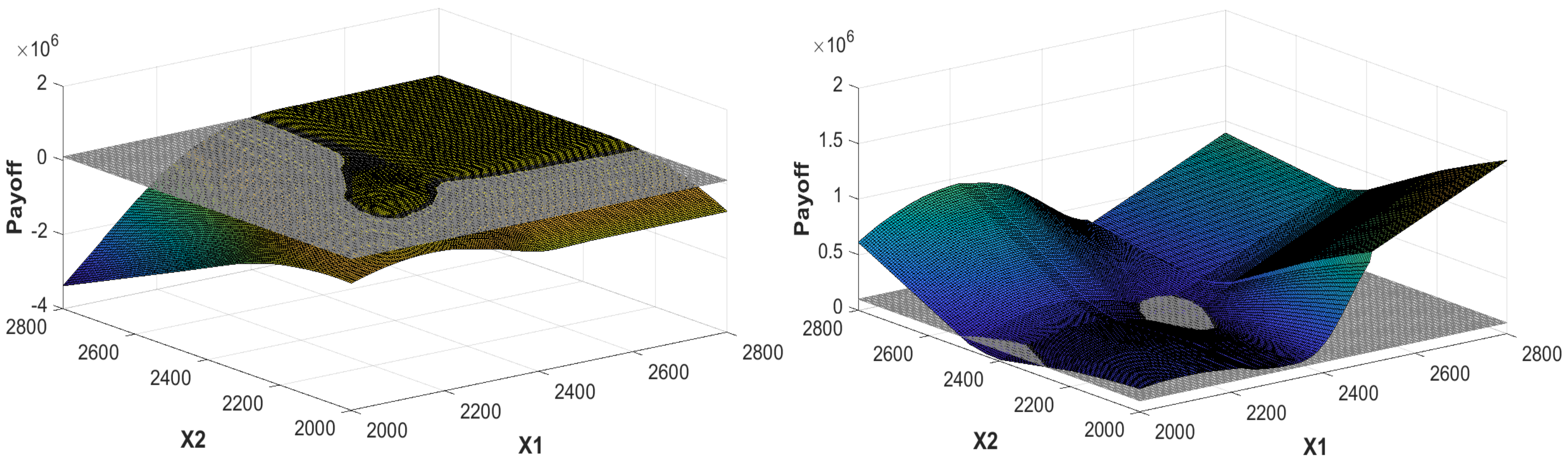}
    \caption{The  payout of the optimal portfolios by semi-static optimization as functions of $X_1$ and $X_2$, obtained with $\sigma=0.08$ (left) and $\sigma=0.2$ (right). The grey horizontal planes represent the initial wealth.}\label{fig:optPayoffDifferentSigmaSSP}
\end{figure}

\subsection{Arbitrage}\label{sec:arb}

We found that, with the quotes obtained from Bloomberg, there exists an arbitrage opportunity if the index can be traded without transaction costs. Due to the finite quantities at the best bid and ask quotes, however, the optimization model admits a bounded solution so that the pricing and hedging problems still make economic sense. 

To identify an arbitrage portfolio, we add the constraint 

\[
\sum_{t=1}^T \sum_{j\in J}p_t^j x^j+\sum_{t=1}^{T-1}\Delta X_{t+1} z_t\geq w \quad P\text{-a.s.}
\]
to problem \eqref{alm}. This means that the portfolio payout is at least the initial wealth in all scenarios. In numerical computations, we impose the constraint on all quadrature points. As the payout is a linear function between the strikes, the constraint will then hold everywhere. Figure \ref{fig:arbitragePortSSP} represents the structure of the arbitrage strategy and Figure~\ref{fig:arbitragePayoffSSP} plots the corresponding payout. The solution uses both static and dynamic trading and achieves a net payout that never falls below the initial wealth but is likely to end up strictly higher.

\begin{figure}[H]
  \centering
  \includegraphics[width=0.9\textwidth,height=0.45\textwidth]{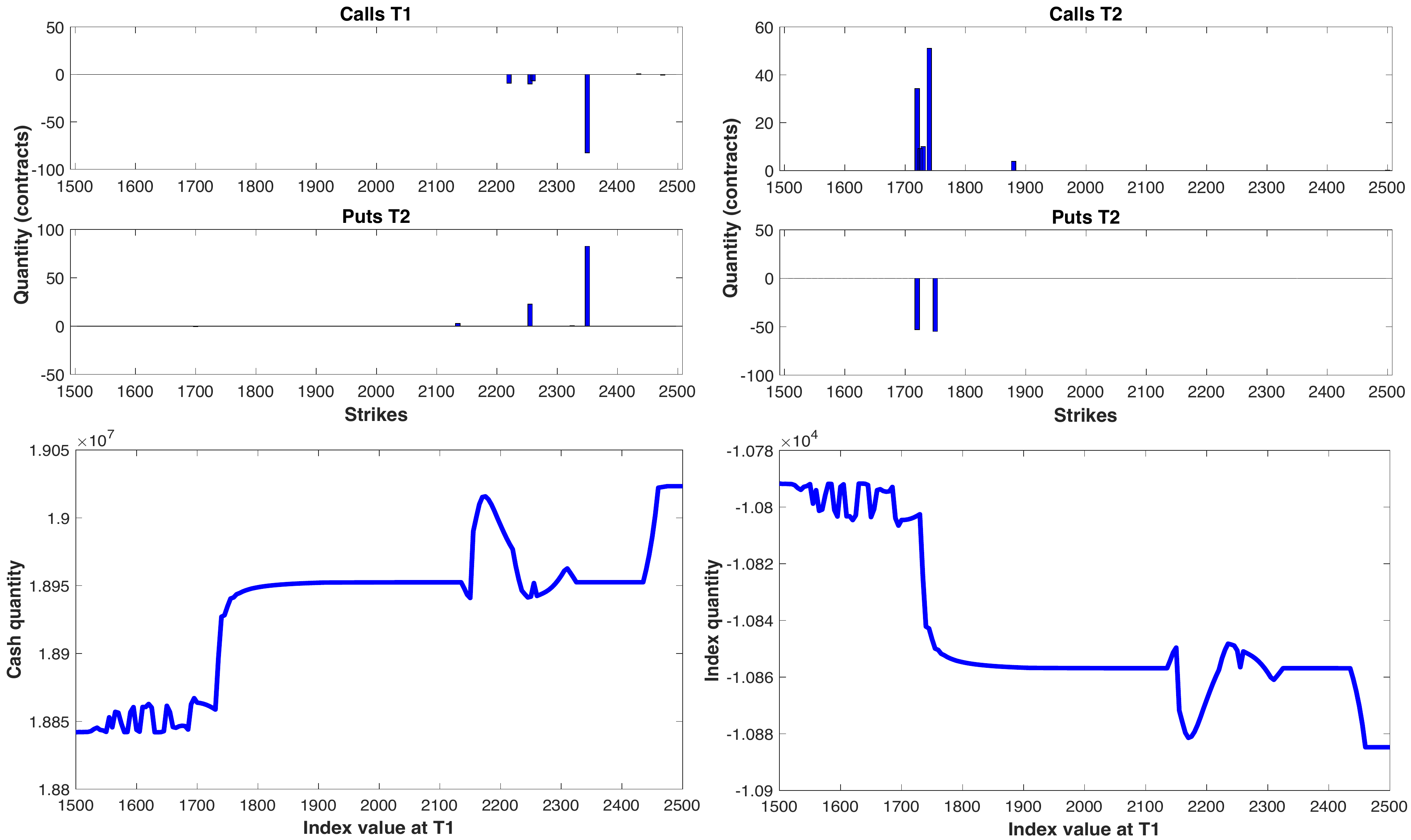}
  \caption{The structure of the arbitrage strategy}\label{fig:arbitragePortSSP}
\end{figure}

\begin{figure}[H]
  \centering
  \includegraphics[width=0.5\textwidth,height=0.32\textwidth]{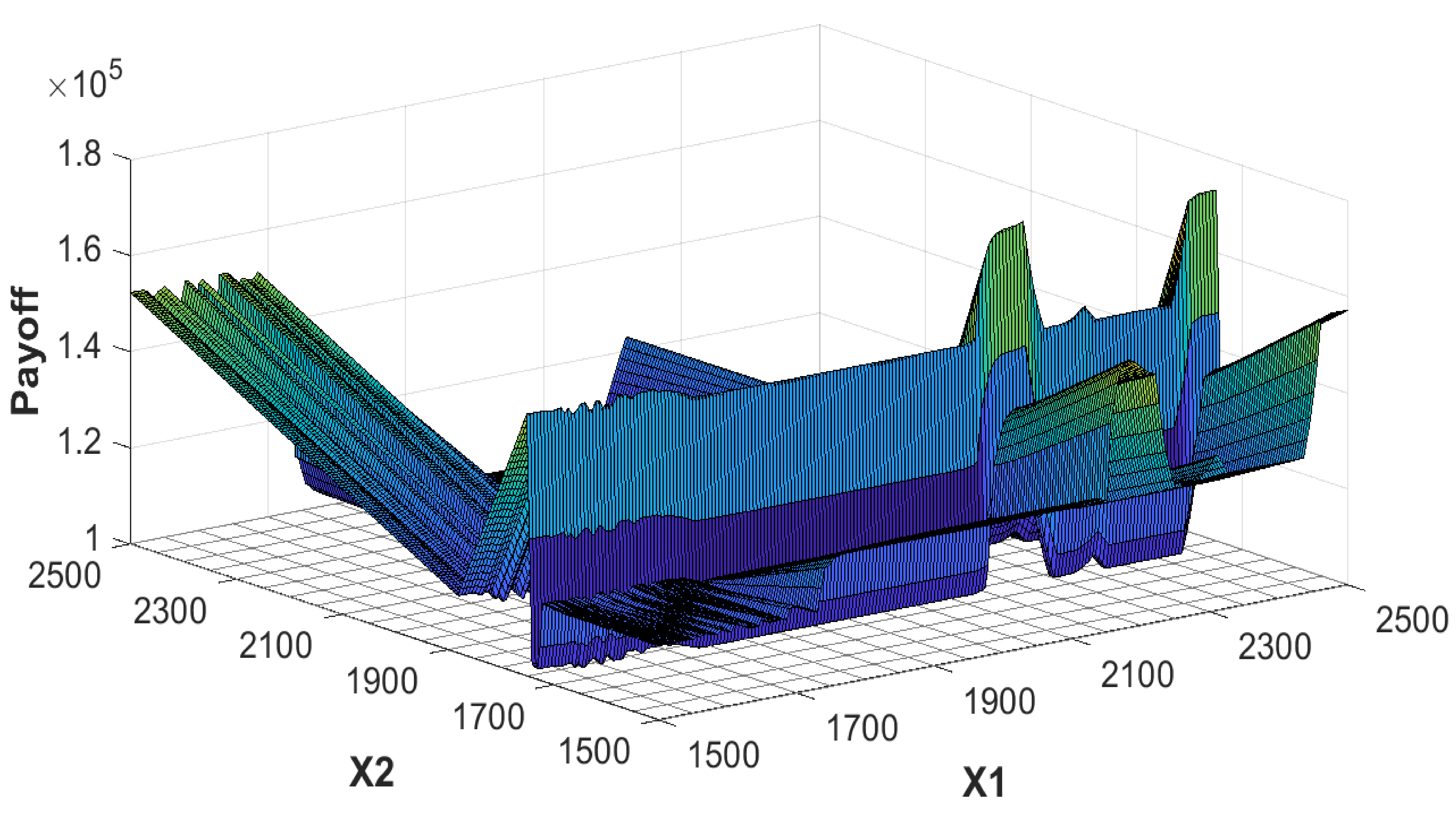}
    \caption{The  payout of the arbitrage portfolio as a function of $X_1$ and $X_2$}\label{fig:arbitragePayoffSSP}
\end{figure}

\subsection{Semi-static problem with transaction costs}

We will now study the effect of a bid-ask spread on the dynamically traded underlying. Figure~\ref{fig:optPortBaseSSPSpread} illustrates the structure of the optimal solution when the proportional transaction cost is $\delta=0.1$. The optimized options portfolio is sparser than the one obtained with perfectly liquid underlying. In addition, the quantities traded are smaller in the options as well as the index. A kernel density plot of the net payout is given in Figure~\ref{fig:optPayoffBaseSSPSpread}.

\begin{figure}[H]
  \centering
  \includegraphics[width=0.9\textwidth,height=0.45\textwidth]{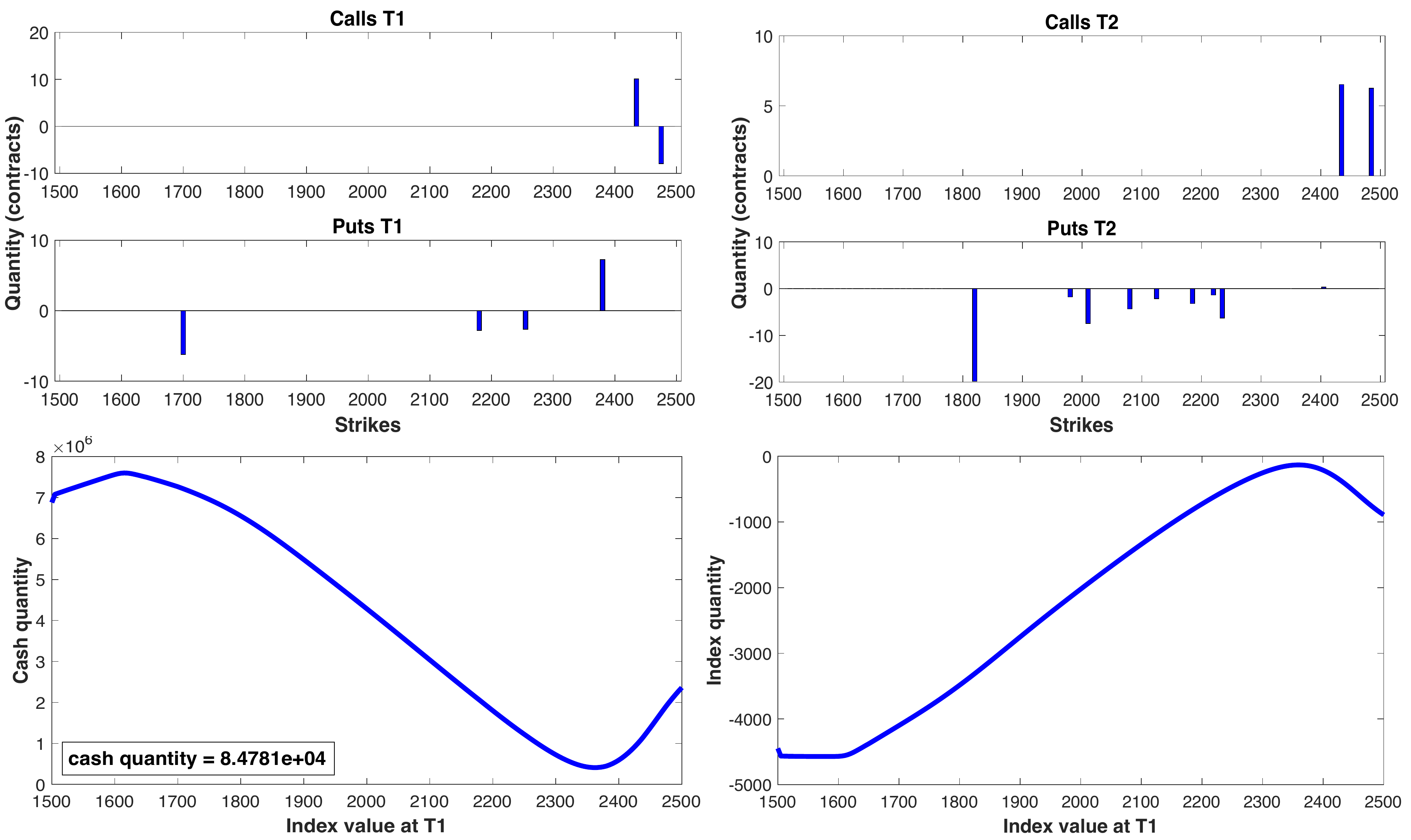}
    \caption{The structure of the optimal portfolio by semi-static optimization with 0.1\% transaction cost}\label{fig:optPortBaseSSPSpread}
\end{figure}
\begin{figure}[H]
  \centering
    \includegraphics[width=0.5\textwidth,height=0.32\textwidth]{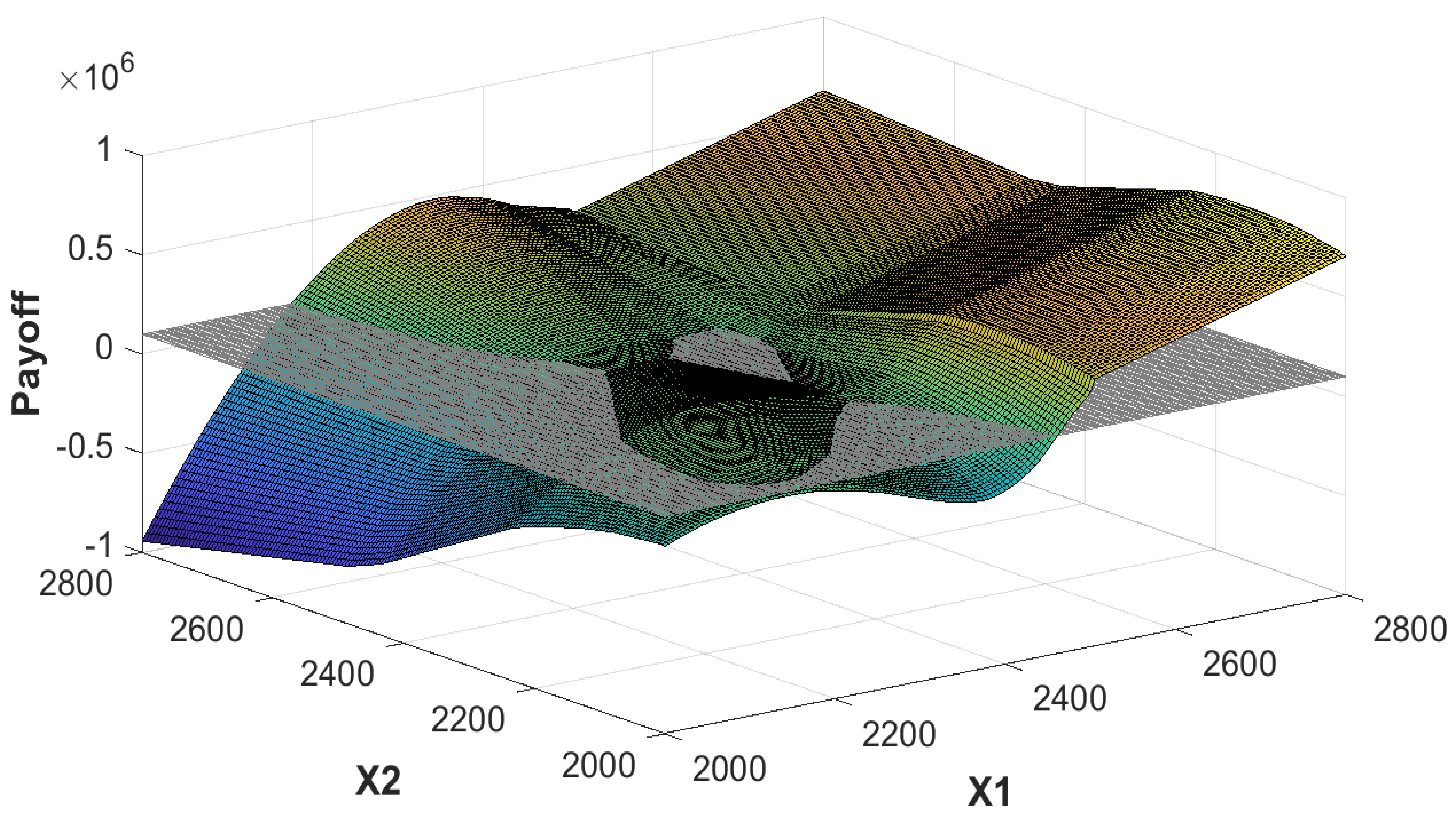}
    \caption{The payout of the optimized portfolio by semi-static optimization as a function of $X_1$ and $X_2$ with 0.1\% transaction cost. The grey horizontal plane represents the initial wealth.}\label{fig:optPayoffBaseSSPSpread}
\end{figure}
To examine the effects of the transaction cost further, we study the payouts of the optimized portfolios for varying levels of the transaction costs. The left plot of Figure \ref{fig:optPayoffDifferentSpreadSSP} shows the payout of the optimized portfolio with 0.1\% transaction cost subtracted by the payout of the optimized portfolio with 1\% transaction cost, whereas the right plot shows  the payout of the optimized portfolio with 0.1\% transaction cost subtracted by the payout of the optimized portfolio with 10\% transaction cost. Figure \ref{fig:optPortIndexDifferentSpread} shows the optimal index quantities bought or sold at $t=1$ obtained with different transaction costs as functions of $X_1$. 

We see that, for the index values up to 2500, which is the highest strike among the options available in the market, a higher transaction cost results in payouts that tend to be higher when $X_1$ and $X_2$ are close to each other. As the transaction cost increases, we invest less in the index at $t=1$; see Figure  \ref{fig:optPortIndexDifferentSpread}.


\begin{figure}[H]
  \centering
    \includegraphics[width=0.85\textwidth,height=0.32\textwidth]{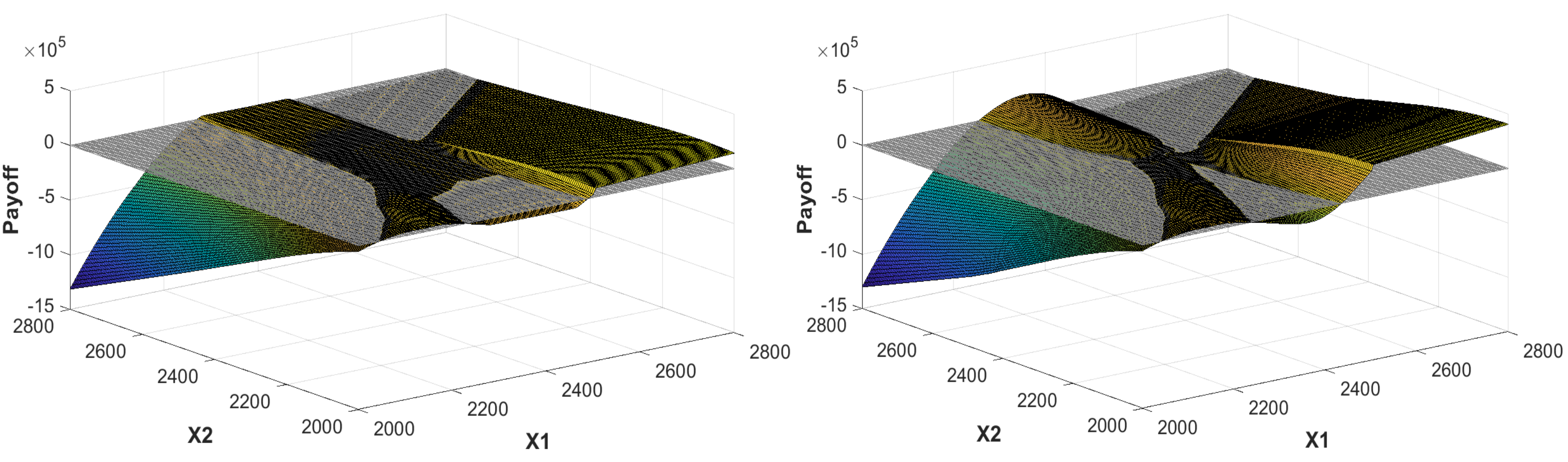}
    \caption{The payout of the optimized portfolio with 0.1\% transaction cost subtracted by 1\% transaction cost\rq s (left). The payout of the optimized portfolio with 0.1\% transaction cost subtracted by 10\% transaction cost\rq s (right). }\label{fig:optPayoffDifferentSpreadSSP}
\end{figure}
\begin{figure}[H]
  \centering
    \includegraphics[width=0.5\textwidth,height=0.32\textwidth]{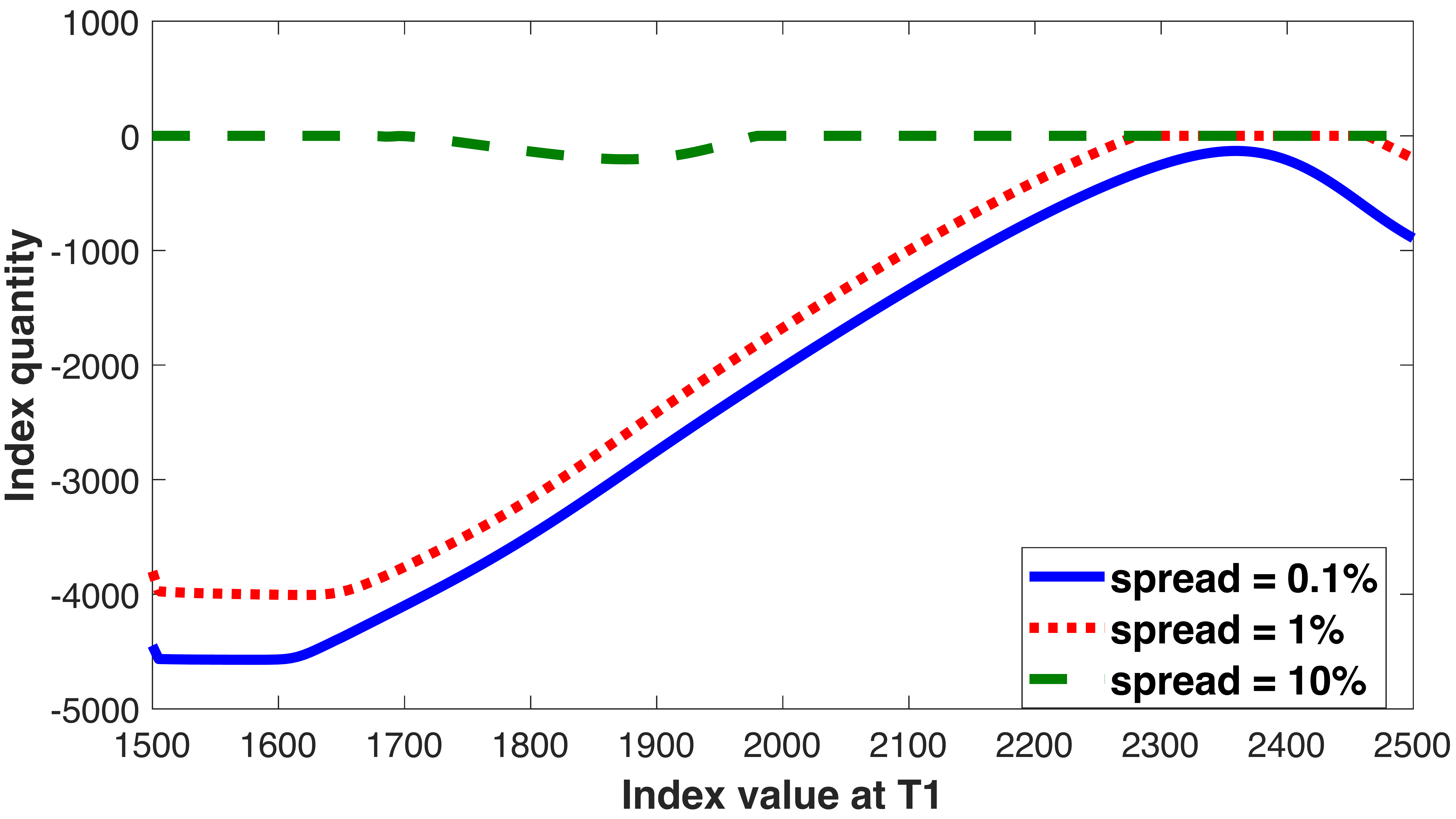}
    \caption{The optimal index quantities bought or sold at $t=1$ by semi-static optimization as functions of $X_1$ with different transaction costs}\label{fig:optPortIndexDifferentSpread}
\end{figure}

Table~\ref{table:disutilitiesSSP-spread} shows how the expected loss function value increases with the transaction cost.

\begin{table}[h!]
  \centering
  \resizebox{0.75\textwidth}{!}{  
  \begin{tabular}{l|c|c|c|c}
\hline
  transaction cost& 0\%& 0.1\%& 1\% &10\%\\
\hline
   log objective& -2.4404 &  -2.3845 & -2.3559 & -2.3136 \\
\hline
   \end{tabular}
   }
  \captionsetup{width=1\textwidth}
  \caption{Logarithms of the optimum objective values in semi-static optimization with different transaction costs}
  \label{table:disutilitiesSSP-spread}
\end{table}

\subsection{Dynamic trading without options}

To illustrate the benefits of employing buy-and-hold strategies in the quoted options, we will compare the results with a purely dynamic optimization model where we are not allowed to trade the options. Other than that, the model is identical with the ones studied above. The numerical optimization is done with the Galerkin discretization and quadrature approximations as described in Section~\ref{sec:numopt}.

Figure~\ref{fig:optimizationVG} plots the mark-to-market values of the optimal index holding at time $t=1$ as functions of the underlying $X_1$ for varying levels of transaction costs $\delta$. When there are no transaction costs, the amount of wealth invested in the underlying does not depend on the value of the underlying except for the more extreme values. This is in line with the theory which says that with exponential utility, the amount of wealth invested in the risky assets is constant. The deviations at the extremes are due to discretization errors. The  corresponding terminal wealth of the optimal index position as a function of $X_1$ and $X_2$ when there are no transaction costs is given in Figure \ref{fig:optDynamicsPayoff}. With higher transaction costs, the terminal wealth is lower as one would expect, and with $0.20\%$ transaction cost, the terminal wealth is constant as there is no trading at all. 

\begin{figure}[H]
  \centering
    \includegraphics[width=0.65\textwidth,height=0.42\textwidth]{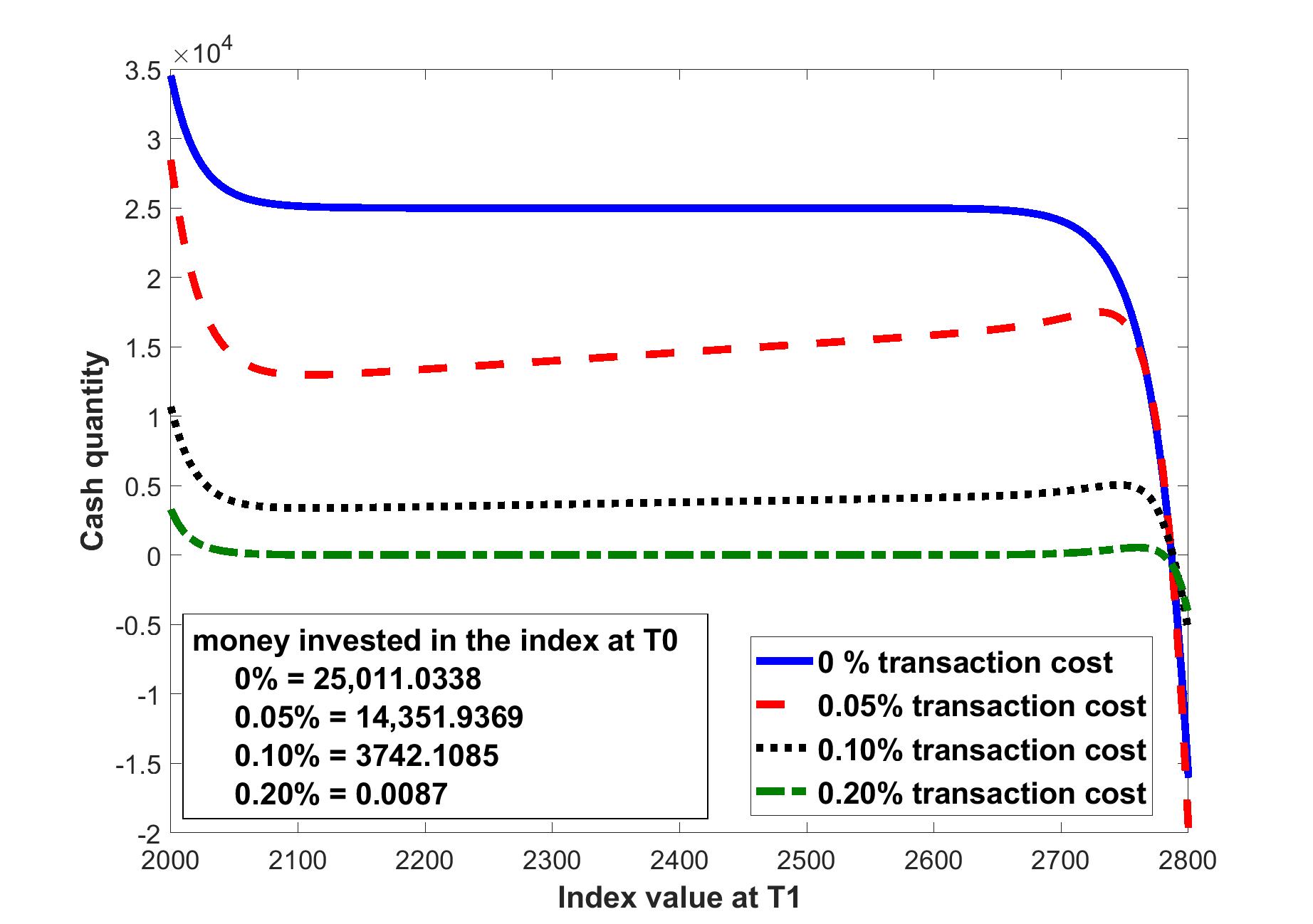}
    \caption{Mark-to-market value of the index at time $t=1$ as a function of $X_1$ obtained by dynamic optimization with different transaction costs }\label{fig:optimizationVG}
\end{figure}

\begin{figure}[H]
  \centering
    \includegraphics[width=0.5\textwidth,height=0.32\textwidth]{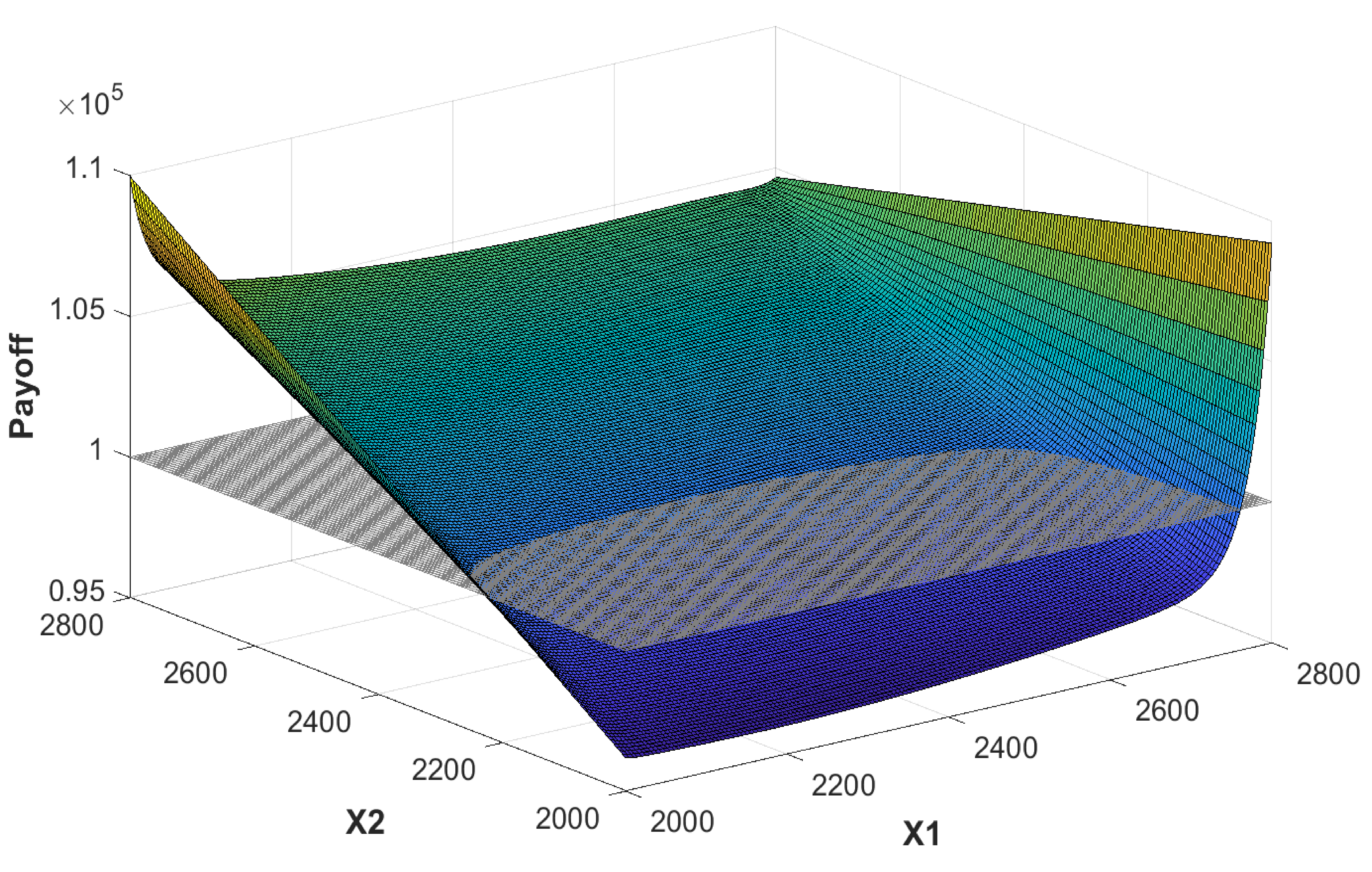}
    \caption{The payout of the optimized portfolio by purely dynamic optimization when there are no transaction costs. The grey horizontal plane represents the initial wealth.}\label{fig:optDynamicsPayoff}
\end{figure}

\section{Indifference pricing of path-dependent options}\label{sec:idpnum}

To illustrate numerical indifference pricing, we will consider a standard call option and four path-dependent options namely, a knock-out call option with payoff 
\begin{equation}
c(X_{1},X_{2})=
\begin{cases}
  (X_{2}-K)^+ & \text{if $X_{1}< B$},\\
  0 & \text{if $X_{1}\ge B$,}
\end{cases}
\end{equation}
an Asian call option with payoff
\begin{equation}
c(X_{1},X_{2})=\bigg(\frac{X_{1}+X_{2}}{2}-K\bigg)^+,
\end{equation}
a look-back call option with payoff 
\begin{equation}
c(X_{1},X_{2})=\max_{t=1,2} \{(X_{t}-K)^+\},
\end{equation}
and a look-back digital call option with payoff
\begin{equation}
c(X_{1},X_{2})=
\begin{cases}
  10 & \text{if $X_{1}\ \text{or}\ X_{2} \ge K $},\\
  0 & \text{otherwise,}
\end{cases}
\end{equation}
all with strike $K=2,350$ and barrier $B=2,400$. The contract size for all options is 100. 


Tables~\ref{table:indiffPrices}--\ref{table:indiffPricesDynamics} give the indifference prices for buying and selling as well as the super- and subhedging costs for the above options. The prices in Table~\ref{table:indiffPrices} were obtained with semi-static hedging with perfectly liquid underlying while those in Table~\ref{table:indiffPricesTransactionCosts} were obtained with transaction cost of $\delta=0.1\%$. Table~\ref{table:indiffPricesDynamics} gives the prices obtained without statically held options. To make the pricing of the call option nontrivial, the call option with strike 2,350 is taken out from the hedging instruments when being priced. The super- and subhedging costs are computed on the interval $[1000,3000]$ instead of all positive real values. This is due to the fact that some options cannot be super- or subhedged for all possible real values resulting in super- and subhedging  costs being infinity.

As expected, the indifference prices for the vanilla call option are more expensive than the knock-out option but cheaper than the look-back call option. The indifference prices for the Asian call option are cheaper than the vanilla call option. This is not surprising as $X_{2}$ is likely to deviate from $X_{0}$ more than from $X_{1}$. The indifference prices for the look-back call option are between 0 and 10 but closer to 10 as it is in the money.

As reported in Section~\ref{sec:arb}, there is arbitrage opportunity in the semi-static model without transaction costs. Accordingly, the superhedging costs are below the subhedging costs. However, the existence of the arbitrage does not prohibit us from computing sensible indifference prices. One should note that in the presence of arbitrage, the quantity constraints for the options are binding so Theorem~\ref{thm:ab} does not apply in the present situation. Adding a 0.1\% transaction cost on the underlying removes the arbitrage and puts us back in the setting of Theorem~\ref{thm:ab} in terms of the order of the four prices; see Table~\ref{table:indiffPricesTransactionCosts}.

As expected, adding transaction costs increases superhedging costs and lowers the subhedging costs. Removing the statically traded options has a similar effect. This is simply because the construction of a superhedging strategy becomes cheaper when trading costs are reduced. The same does not apply to indifference pricing because both sides of the indifference inequality increase when trading becomes more expensive.

Without the statically traded options, the true superhedging cost is $+\infty$ for all but the digital option. Accordingly, the numerically computed superhedging costs in Table~\ref{table:indiffPricesDynamics} would converge to infinity when the scenario grid is extended further. Similarly, the true subhedging costs of all but the call and Asian option are zero.

\begin{table}[H]
  \centering
  \resizebox{0.9\textwidth}{!}{  
    \begin{tabular}{|l|c|c|c|c|} 
      \hline
      claim & subhedging& buying price & selling price&superhedging\\ 
      \hline
      call   &52.9626&  45.3296& 45.3939&   37.4974\\ 
      knock-out call &18.1167&22.4763&22.7125&18.6974\\ 
      Asian  &38.9026& 35.0562& 35.1019& 29.8066 \\ 
      look-back call &53.6604& 53.9293&54.0110& 51.5058\\ 
      look-back digital &14.4026&7.6834& 7.6966& 0.6058\\ 
      \hline
    \end{tabular}
  }
  \captionsetup{width=1\textwidth}
  \caption{Indifference prices, together with super- and subhedging costs by semi-static hedging without transaction costs on the underlying}
  \label{table:indiffPrices}
\end{table}
\begin{table}[H]
  \centering
  \resizebox{0.9\textwidth}{!}{  
    \begin{tabular}{|l|c|c|c|c|} 
      \hline
      claim & subhedging& buying price & selling price& superhedging\\ 
      \hline
      call & 43.4250  &  44.5308& 44.8265&45.8000   \\ 
      knock-out call&4.8957 &20.6770&21.1444&29.5397\\ 
      Asian &29.8327 & 35.0303& 35.2427 & 38.9201\\ 
      look-back call & 43.9763& 53.7226&54.0400& 61.1320 \\ 
      look-back digital&5.2640 &7.5498& 7.5663 &9.2374\\ 
      \hline
    \end{tabular}
  }
  \captionsetup{width=1\textwidth}
  \caption{Indifference prices, together with super- and subhedging costs  by semi-static hedging with  0.1\% transaction cost on the underlying}
  \label{table:indiffPricesTransactionCosts}
\end{table}
\begin{table}[H]
  \centering
  \resizebox{0.9\textwidth}{!}{  
    \begin{tabular}{|l|c|c|c|c|} 
      \hline
      claim & subhedging& buying price & selling price& superhedging\\ 
      \hline
      call & 10.0000 & 49.9490 & 51.2605 & 442.0000  \\ 
      knock-out call& 0.0000&15.3326&16.5480&442.0000 \\ 
      Asian & 0.0000& 41.1187& 42.1857& 442.0000\\ 
      look-back call & 10.0000& 60.4879 &62.3530& 485.2906  \\ 
      look-back digital &0.1538 &6.4321& 6.4469  &10.0000\\ 
      \hline
    \end{tabular}
  }
  \captionsetup{width=1\textwidth}
  \caption{Indifference prices, together with super- and subhedging costs by two-period dynamic hedging without statically held options}
  \label{table:indiffPricesDynamics}
\end{table}

The hedging portfolios, which are $x-\bar{x}$ and $z-\bar{z}$ where $\bar{x}$ and $x$ are options portfolios, and  $\bar{z}$ and $z$ are  index quantities before and after selling the options, as well as their payouts for each hedging will be shown in the later subsections.

\subsection{Optimal hedges}\label{sssec:ssliquid}

Figure~\ref{fig:callSSP} illustrates the hedging strategy for selling one contract of the call option strike $K=2,350$. This includes the hedging portfolio, as well as its payout plotted together with the payoff (grids) of the call option. The payouts of the hedging portfolios for the other options are shown in Figure~\ref{fig:payoffHedgingSemiStatic}. We see that the path-dependent exotic options can be hedged reasonably well especially for the scenarios with higher probability of occurring.

\begin{figure}[H]
  \centering
    \includegraphics[width=0.85\textwidth,height=0.57\textwidth]{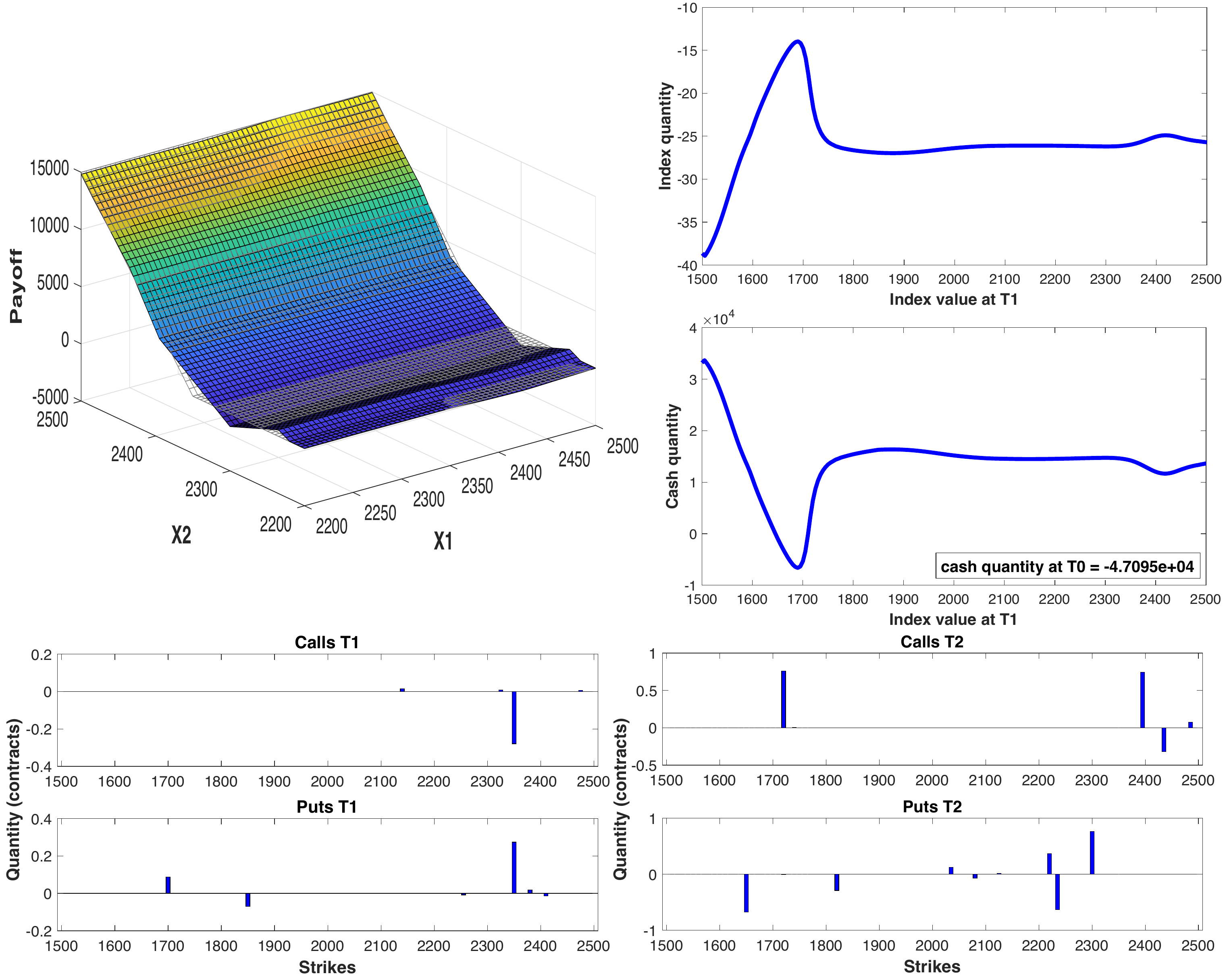}
    \caption{The hedging portfolio for selling one contract of a call option with strike 2,350.}\label{fig:callSSP}
\end{figure}

\begin{figure}[H]
  \centering
    \includegraphics[width=0.85\textwidth,height=0.57\textwidth]{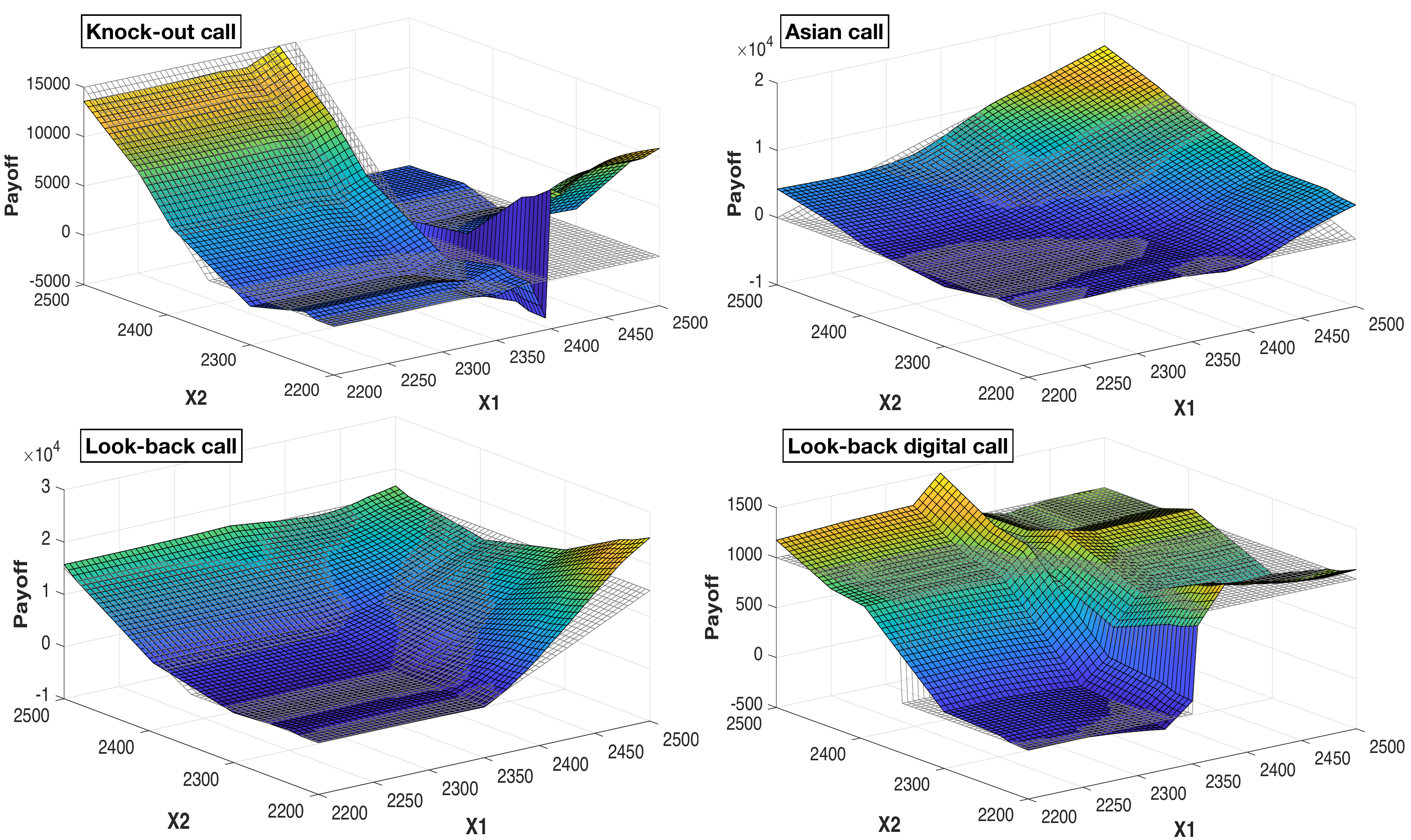}
    \caption{The payouts of the hedging portfolios for selling one contract of a knockout call option, Asian call option, look-back call option, and look-back digital option all with strike 2,350 and barrier 2,400.}\label{fig:payoffHedgingSemiStatic}
\end{figure}

\subsection{Optimal hedges with transaction costs}\label{sssec:ssill}

Figure~\ref{fig:callSSP-spread01} illustrates the hedging strategy for selling one contract of the call option strike $K=2,350$ when the underlying is subject to 0.1\% transaction cost. The payouts of the hedging portfolios for the other options are shown in Figure~\ref{fig:payoffHedgingSemiStatic}. 

\begin{figure}[H]
  \centering
    \includegraphics[width=0.85\textwidth,height=0.57\textwidth]{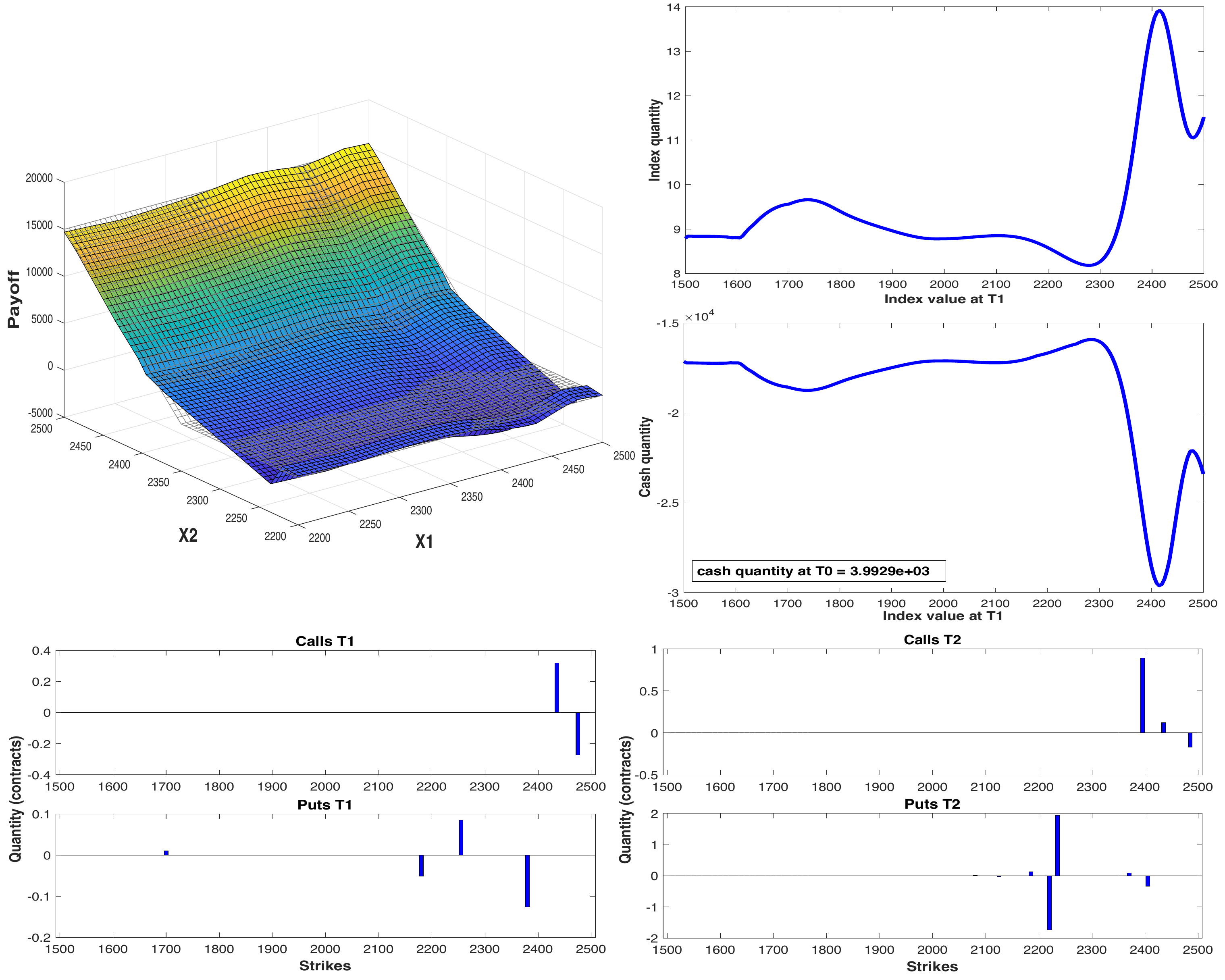}
    \caption{The hedging portfolio for selling one contract of a call option with strike 2,350 with 0.01\% transaction cost.}\label{fig:callSSP-spread01}
\end{figure}

\begin{figure}[H]
  \centering
    \includegraphics[width=0.85\textwidth,height=0.57\textwidth]{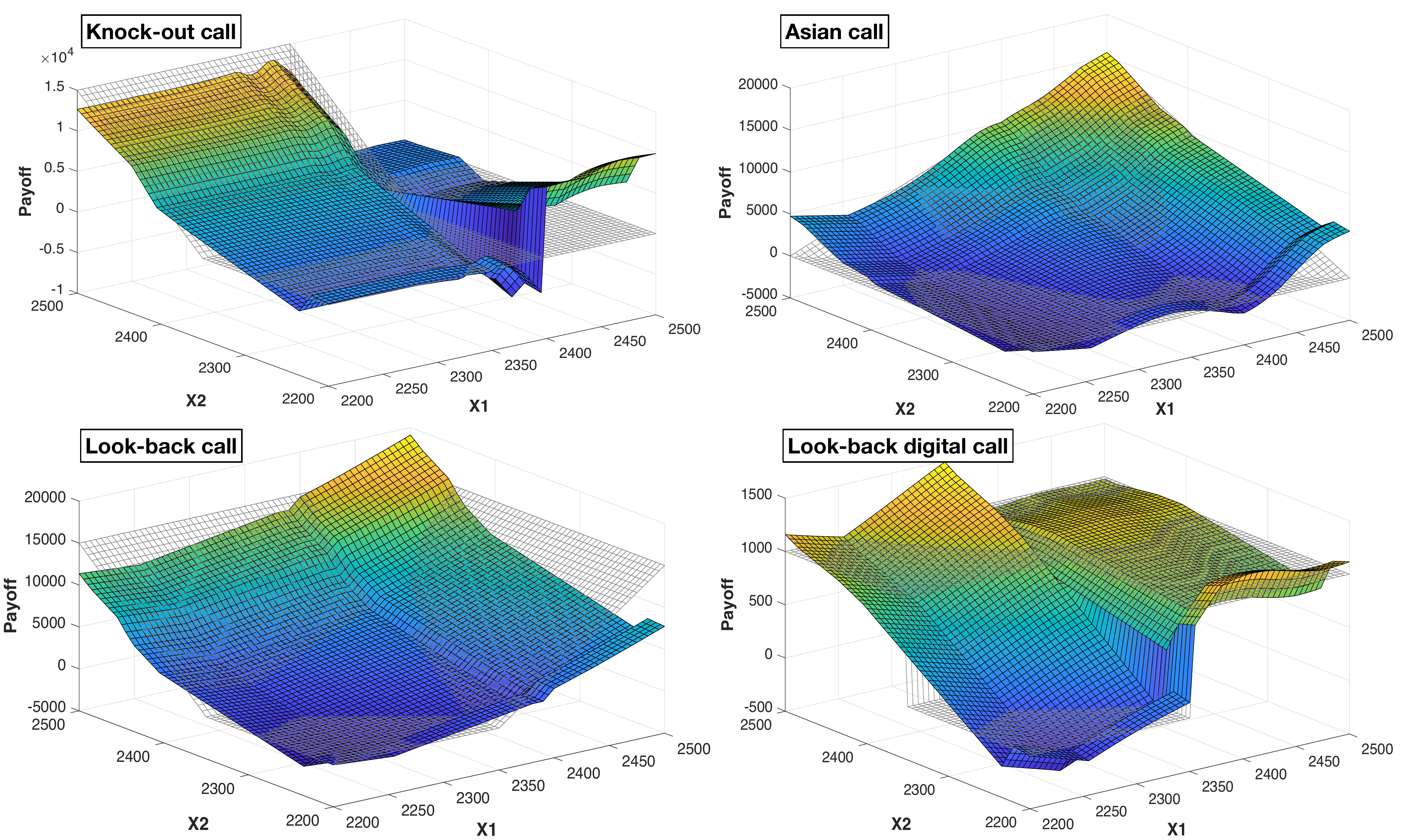}
    \caption{The payouts of the hedging portfolios for selling one contract of a knockout call option, Asian call option, look-back call option, and look-back digital option all with strike 2,350 with 0.01\% transaction cost.}\label{fig:payoffHedgingSSP-spread01}
\end{figure}

Despite having the transaction cost, the path-dependent options are still hedged well. However, the lower the transaction cost, the better the hedge as we can see from Figure \ref{fig:lookBackCallDifferentTransaction} which shows the payouts of the hedging portfolios for selling the look-back call option with strike 2,350 when the transaction costs are 1 and 10 percents. The semi-static hedging with the 10\% transaction cost coincides with the static hedging as the underlying is not traded. Note that static hedging is a special case of semi-static hedging. We see that the 2-dimensional shapes of the payout are identical for any given values of $X_1$ or $X_2$




\begin{figure}[H]
  \centering
    \includegraphics[width=0.85\textwidth,height=0.32\textwidth]{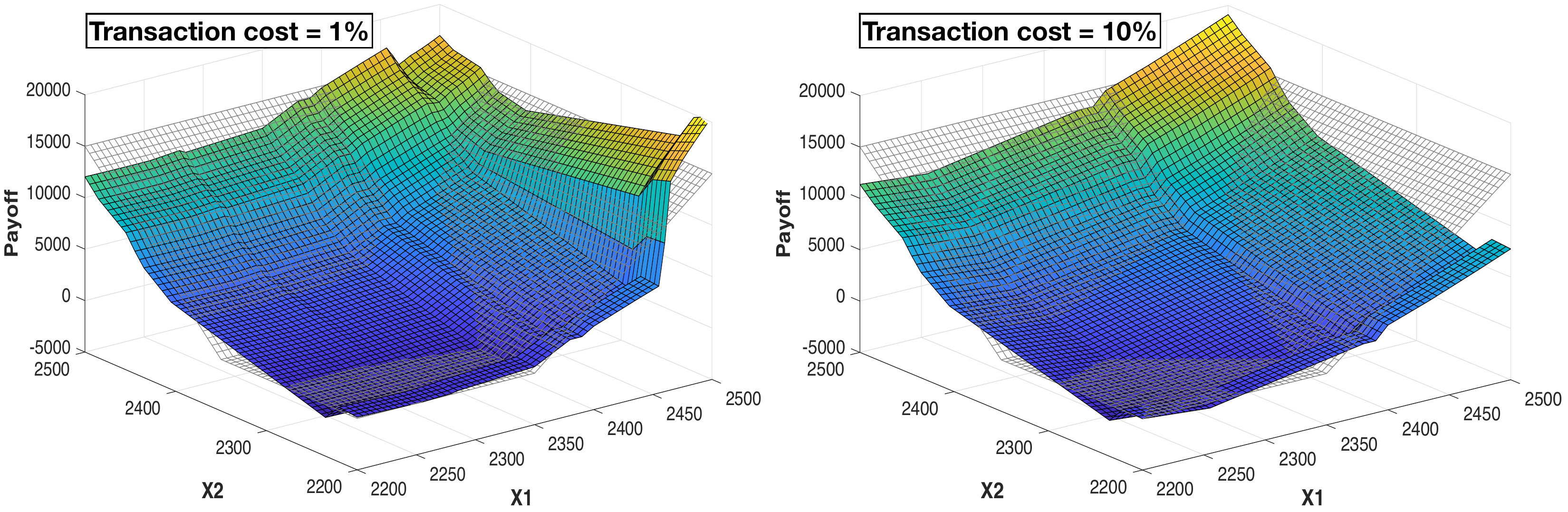}
    \caption{The payouts of hedging portfolios for selling one contract of a lookback call option with strike 2350 with 1\% and 10\% transaction costs.}\label{fig:lookBackCallDifferentTransaction}
\end{figure}

%
%
%
%
%

\subsection{Optimal hedging without options}\label{sssec:dyn}

Figure~\ref{fig:callDynamics} illustrates the hedging strategy for selling one contract of the call option strike $K=2,350$ by two-period dynamic hedging. The payouts of the hedging portfolios for the other options are shown in Figure~\ref{fig:pricingDynamic}. Only cash and the underlying, allowed to be traded without transaction costs at $t=0,1$, are the hedging instruments. We see that, without the call and put options as the hedging instruments, they badly hedge the options. However, hedging portfolios tend to be more profitable when $X_1$ and $X_2$ are close to each other which is an area with higher probability of occurring.

\begin{figure}[H]
  \centering
    \includegraphics[width=0.85\textwidth,height=0.32\textwidth]{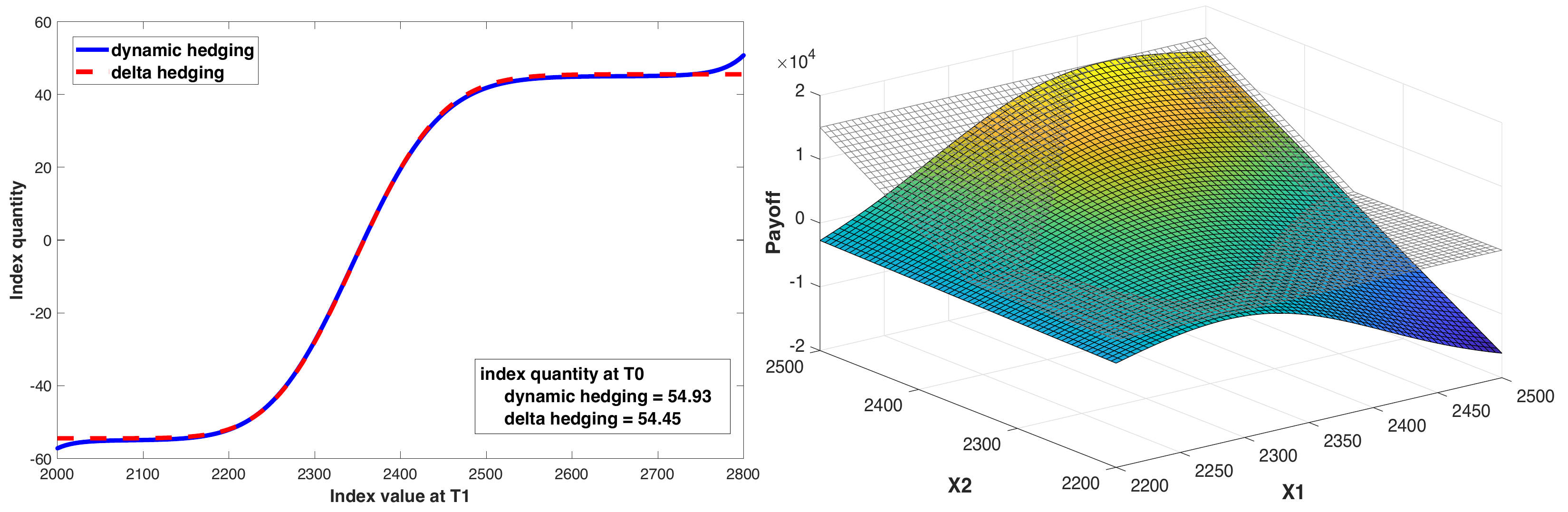}
    \caption{The quantity of the index bought at $t=1$ as a function of $X_1$ of the hedging portfolio for selling one contract of a call option with strike 2,350 by dynamic strategy (left) and its payout (right).}\label{fig:callDynamics}
\end{figure}

\begin{figure}[H]
  \centering
    \includegraphics[width=0.85\textwidth,height=0.57\textwidth]{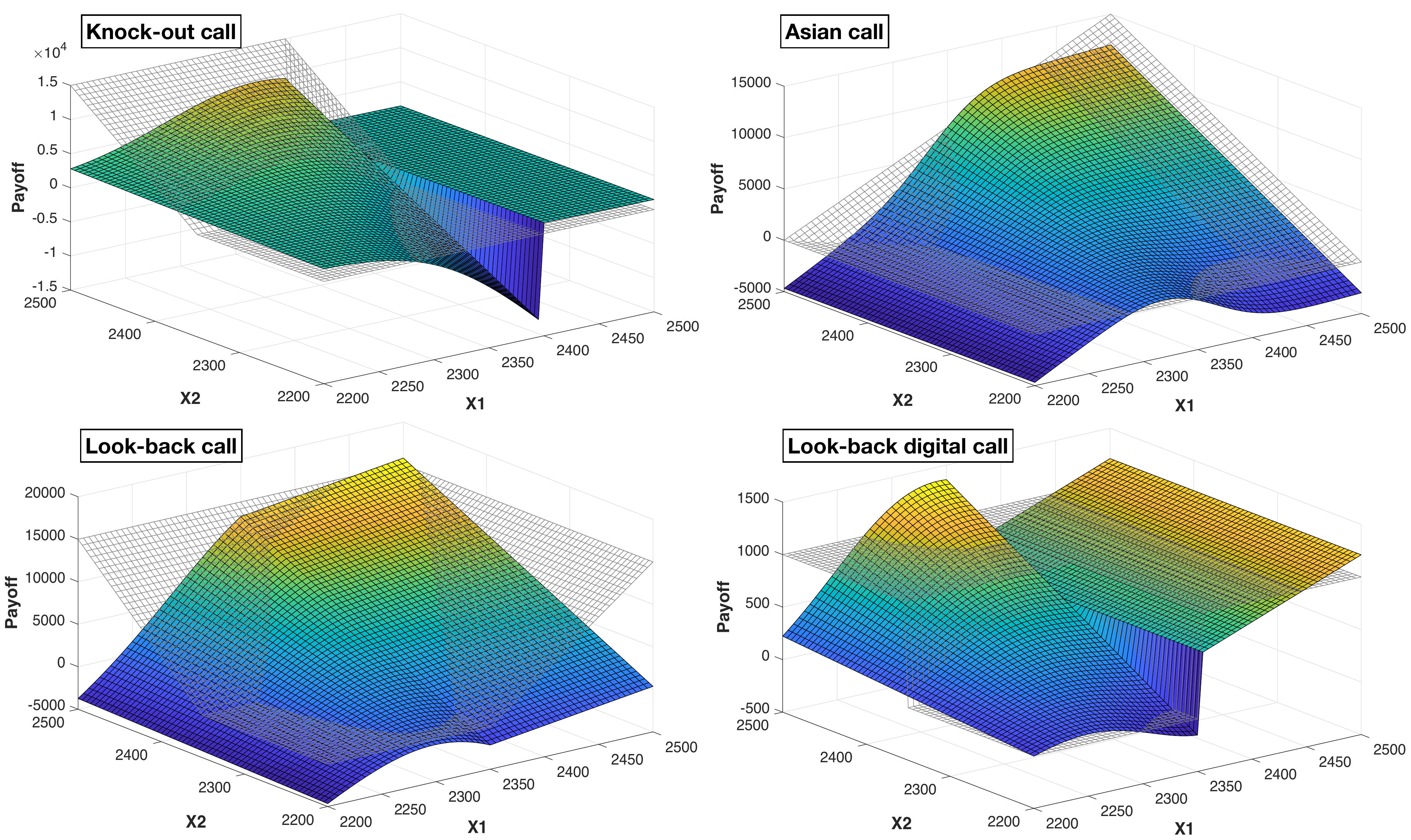}
    \caption{The payouts of the hedging portfolios for selling one contract of a knockout call option, Asian call option, look-back call option, and look-back digital option all with strike 2,350 and barrier 2,400 by dynamic hedging.}\label{fig:pricingDynamic}
\end{figure}

The quantity of trade in the index as shown in Figure \ref{fig:callDynamics} looks similar to the one of the delta hedging. This is very surprising because the dynamic hedging is implemented in a two-period setting, whereas the delta hedging is a continuous trading strategy. The indifference prices for buying and selling are 49.9490 and 51.2605, whereas the Black-Scholes price with $\mu=0$ and $\sigma=0.1206$ is 50.7242. Indifference prices depend on an agent's initial position, risk preference and market liquidity, while the Black-Scholes price does not.

Figure \ref{fig:deltaVGSpreadComparison} shows the quantities of the index traded at $t=1$ of the hedging portfolios for selling one contract of a call option with strike 2,350 as  functions of $X_1$ obtained by two-period dynamic hedging with different transaction costs. We see that the quantities traded in the underlying at both $t=0$ and $t=1$ decrease as the transaction cost increases. Except at the tails, which have low probability of occurring, the strategies are to buy some underlying at the beginning and buy more if it goes up or sell if it goes down. However, we see that, for the $0.2\%$ transaction cost, it is optimal to do nothing at $t=1$ if the index value at $t=1$ is below 2400.

\begin{figure}[H]
\centering
    \includegraphics[width=0.5\textwidth,height=0.32\textwidth]{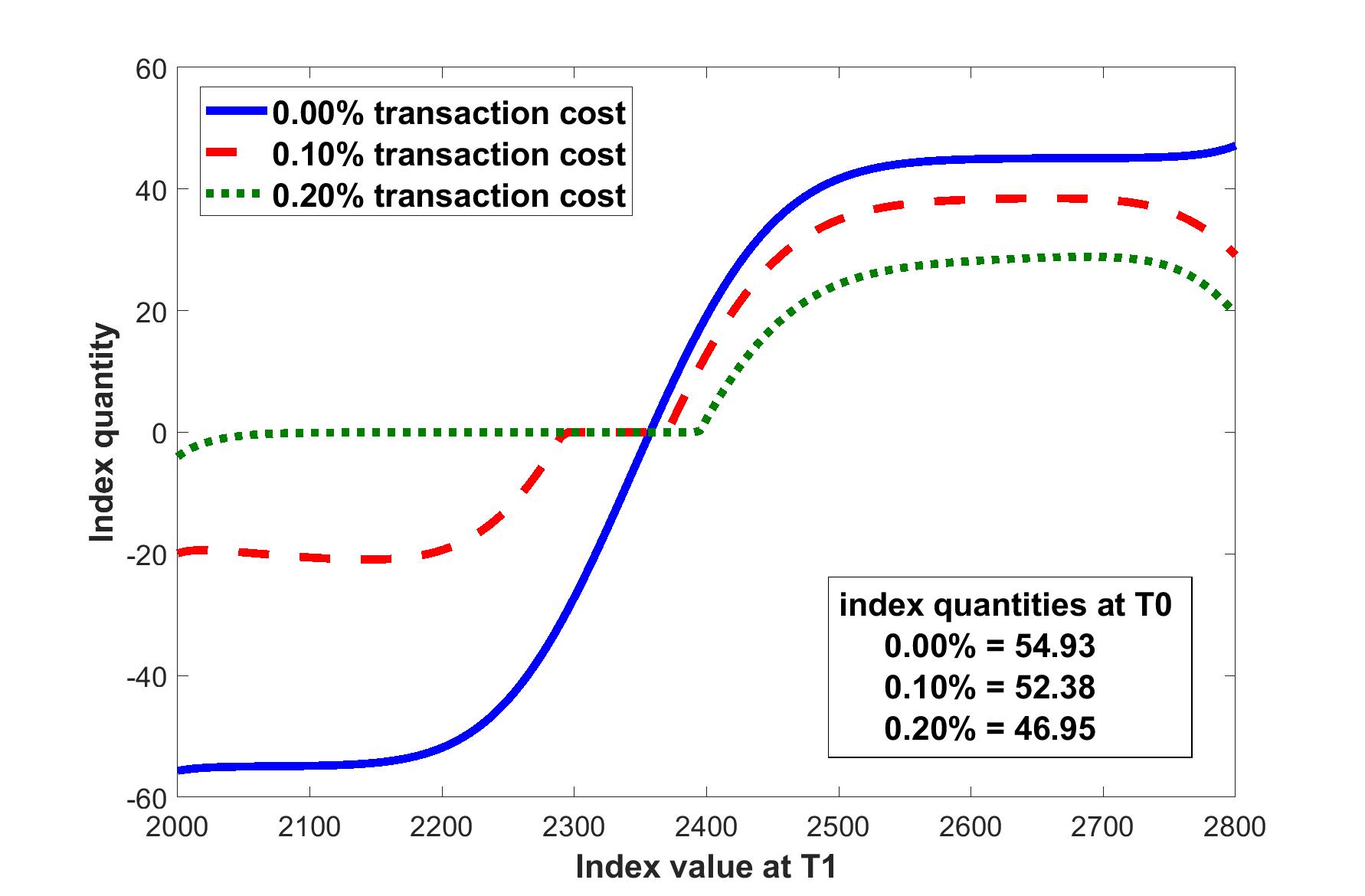}
  \caption{The quantities of the index traded at $t=1$ of the hedging portfolios for selling one contract of a call option with strike 2,350 as functions of $X_1$ by two-period dynamic hedging with different transaction costs.}\label{fig:deltaVGSpreadComparison}
\end{figure}

\bibliographystyle{alpha}
\bibliography{sp}

\begin{thebibliography}{BHLP13}

\bibitem[APR18]{apr18}
J.~Armstrong, T.~Pennanen, and U.~Rakwongwan.
\newblock Pricing index options by static hedging under finite liquidity.
\newblock {\em International Journal of Theoretical and Applied Finance},
  21(06):1850044, 2018.

\bibitem[ApS15]{mosek}
MOSEK ApS.
\newblock {\em The MOSEK optimization toolbox for MATLAB manual. Version 7.1
  (Revision 28)}, 2015.

\bibitem[BHLP13]{bhp13}
M.~Beiglb\"ock, P.~Henry-Labord\`ere, and F.~Penkner.
\newblock Model-independent bounds for option prices---a mass transport
  approach.
\newblock {\em Finance Stoch.}, 17(3):477--501, 2013.

\bibitem[B{\"u}h70]{buh70}
H.~B{\"u}hlmann.
\newblock {\em Mathematical methods in risk theory}.
\newblock Die Grundlehren der mathematischen Wissenschaften, Band 172.
  Springer-Verlag, New York, 1970.

\bibitem[Car09]{car9}
R.~Carmona, editor.
\newblock {\em Indifference pricing: theory and applications}.
\newblock Princeton series in financial engineering. Princeton University
  Press, Princeton, NJ, 2009.

\bibitem[FS11]{fs11}
H.~F{\"o}llmer and A.~Schied.
\newblock {\em Stochastic finance}.
\newblock Walter de Gruyter \& Co., Berlin, extended edition, 2011.
\newblock An introduction in discrete time.

\bibitem[GOo19]{go19}
G.~Guo and J.~Ob\l~\'{o}j.
\newblock Computational methods for martingale optimal transport problems.
\newblock {\em Ann. Appl. Probab.}, 29(6):3311--3347, 2019.

\bibitem[HN89]{hn89}
S.~D. Hodges and A.~Neuberger.
\newblock Optimal replication of contingent claims under transaction costs.
\newblock {\em Reviev of Futures Markets}, 8:222--239, 1989.

\bibitem[IJS04]{ijs04}
A.~Ilhan, M.~Jonsson, and R.~Sircar.
\newblock Portfolio optimization with derivatives and indifference pricing.
\newblock In R.~Carmona, editor, {\em Indifference pricing -theory and
  applications}, pages 183--210. Princeton University Press, 2004.

\bibitem[IJS08]{ijs08}
A.~Ilhan, M.~Jonsson, and R.~Sircar.
\newblock Optimal static-dynamic hedges for exotic options under convex risk
  measures.
\newblock {\em SSRN Electronic Journal}, 2008.

\bibitem[IS06]{is06}
A.~Ilhan and R.~Sircar.
\newblock Optimal static-dynamic hedges for barrier options.
\newblock {\em Mathematical Finance}, 16(2):359--385, 2006.

\bibitem[MCC98]{mcc98}
D.~Madan, P.~Carr, and E.~Chang.
\newblock {The Variance Gamma Process and Option Pricing}.
\newblock {\em Review of Finance}, 2(1):79--105, 04 1998.

\bibitem[MS90]{testt}
D.~Madan and E.~Seneta.
\newblock The variance gamma (v.g.) model for share market returns.
\newblock {\em The Journal of Business}, 63(4):511--524, 1990.

\bibitem[Pen14]{pen14}
T.~Pennanen.
\newblock Optimal investment and contingent claim valuation in illiquid
  markets.
\newblock {\em Finance Stoch.}, 18(4):733--754, 2014.

\bibitem[PP18]{pp18e}
T.~Pennanen and A.-P. Perkki\"{o}.
\newblock Convex duality in optimal investment and contingent claim valuation
  in illiquid markets.
\newblock {\em Finance Stoch.}, 22(4):733--771, 2018.

\bibitem[PP19]{pp19}
T.~Pennanen and A.-P. Perkki\"{o}.
\newblock Convex duality in nonlinear optimal transport.
\newblock {\em J. Funct. Anal.}, 277(4):1029--1060, 2019.

\end{thebibliography}

\end{document}